\documentclass[11pt]{article}
\usepackage{framed}
\usepackage{color}
\usepackage{multirow}
\usepackage{graphicx}
\usepackage{subfigure}
\usepackage{epsfig}
\usepackage{cite}
\usepackage{amsmath, amssymb}
\usepackage{fancybox}
\usepackage{listings}
\usepackage{verbatim}
\usepackage{url}
\usepackage{esint}
\usepackage{fancyhdr}
\usepackage{booktabs}
\usepackage{caption}
\usepackage{dsfont}
\usepackage{amsthm}

\usepackage{mathtools}

\input{epsf.sty}
\newtheorem{remark}{Remark}[section]
\newtheorem{example}{Example}[section]
\newtheorem{proposition}{Proposition}[section]

\newcommand{\sgn}{\mathop{\rm sgn}}

\newcommand{\remove}[1]{}

\mathtoolsset{showonlyrefs}

\title{\bf
A Bertalanffy-Richards growth model perturbed by a time-dependent pattern, statistical analysis
and applications\thanks{This paper is dedicated to the cherished memory of our dear colleague and friend Patricia Rom\'an-Rom\'an.}
}

\author{Antonio \textbf{Di Crescenzo}$^{(1)}$  \and
	Paola \textbf{Paraggio}$^{(2)}$ \and
	Francisco \textbf{Torres-Ruiz}$^{(3)}$$^{(4)}$
	\\
	\normalsize (1)  Dipartimento di Matematica,
	\normalsize Universit\`a di Salerno, 84084 Fisciano (SA), Italy \\
	\normalsize Email:  adicrescenzo@unisa.it \quad ORCID:    0000-0003-4751-7341
	\\
	\normalsize (2)  Dipartimento di Matematica,
	\normalsize Universit\`a di Salerno, 84084 Fisciano (SA), Italy \\
	\normalsize Email:  pparaggio.it@unisa.it \quad ORCID:    0000-0002-3308-7937
	\\
	\normalsize (3) Departamento de Estad\'{\i}stica e I.O., Facultad de Ciencias \\
	\normalsize Universidad de Granada, 18071 Granada, Spain \\
	\normalsize (4) 	Instituto de Matem\'aticas de la Universidad de Granada (IMAG) \\
	\normalsize Calle Ventanilla, 11, 18001 Granada, Spain \\
	\normalsize Email: fdeasis@ugr.es \quad ORCID: 0000-0001-6254-2209
}

\date{}
\begin{document}
	\maketitle
	\begin{abstract}\small
	We analyze a modification of the Richards growth model by introducing a time-dependent perturbation in the growth rate. This modification becomes effective at a special switching time, which represents the first-crossing-time of the Richards growth curve
	 through a given constant boundary. The relevant features of the modified growth model are studied  and compared with those of the original one. A sensitivity analysis on the switching time is also performed. Then, we define two different stochastic processes, i.e.\ a non-homogeneous linear birth-death process and a lognormal diffusion process, such that their means identify to the growth curve under investigation. For the diffusion process, we address the problem of parameters estimation  through the maximum likelihood method. The estimates are obtained via meta-heuristic algorithms (namely, Simulated Annealing and Ant Lion Optimizer). A simulation study to validate the estimation procedure is also presented, together with a real application to oil production in France. Special attention is devoted to the approximation of switching time density, viewed as the first-passage-time density for the lognormal process.
	
\bigskip\noindent
{\em Keywords}: Richards growth model, \and  Non-homogeneous birth-death process, \and Lognormal diffusion process, \and First-passage time, \and Maximum likelihood estimation.
		
\medskip\noindent		
		{\em Mathematical Subject Classification:}  60J70, 
		62M05, 
		60J80. 
	\end{abstract}
	\section{Introduction}
	The Richards growth model is a generalization of the well-known logistic model. The main difference between the two models lays in symmetry properties of the resulting curves. Indeed, the logistic function has a symmetrical pattern with respect to the inflection point, in the sense that the carrying capacity (which is the limit value of the growth function) is twice the value of the curve in the inflection point. This  behavior may not be particularly appropriate to describe some real phenomena which show asymmetrical growth patterns. For this reason, Richards in 1959 introduced a new growth curve that was called Richards growth curve (cf.\ Richards (1959) \cite{Richards1959}). Some authors refer to it as Bertalanffy-Richards growth curve, since Richards extended previous works of Bertalanffy regarding plants growth. The flexibility of the afore-mentioned curve is another great advantage of the model. Indeed, for some particular choices of the involved parameters the most known growth functions (such as the Malthusian, the logistic and the Gompertz) can be obtained from the Richards curve. Nevertheless, if one refers to the classical representation of the curve, it is easy to note that the carrying capacity does not depend explicitly on the initial state. This feature may not be so realistic, since it is intuitive to believe that the maximum achievable value of a population size is influenced by the initial size. \textcolor{blue}{For this reason, a reformulation of the model that includes the initial size in the expression of the carrying capacity is useful}.
	
	The flexibility of the Richards deterministic model is also demonstrated by several applications in different fields which range from agricultural studies (such as in Hiroshima (2007) \cite{Hiroshima2007} and in Gerhard and Moltchanova (2022) \cite{Gerhard2022}), to zoological studies (see for example Matis \textit{et al.}\ (2011) \cite{Matisetal2011}, Nahashon (2006) \cite{Nahashonetal2006}, K\"ohn \textit{et al.}\ (2007) \cite{Kohnetal2007}, Lv \textit{et al.}\ (2007) \cite{Lvetal2007} and Russo \textit{et al.}\ (2009) \cite{Russoetal2009}) and to health sciences (cf.\  Mac\^edo \textit{et al.}\  (2021) \cite{Macedoetal2021}, Smirnova \textit{et al.}\ (2022)  \cite{Smirnova} and Wang \textit{et al.}\ (2012) \cite{Wangetal2012}).
	
	In order to include the randomness (which is typical in the phenomenological reality) in the description of the model, it is useful to introduce a stochastic counterpart of the considered growth curve. Typically, researchers define two classes of stochastic processes related to growth curves, that are the birth-death processes (e.g.\ Asadi \textit{et al.}\ (2020) \cite{Asadietal2020}) whose state-space is given by a discrete set and the diffusion processes (e.g.\ Rom\'an-Rom\'an  \textit{et al.}\ (2018) \cite{RomanTorres2018}) whose state-space is  an interval of the real line. Both processes are usually constructed in such a way that the mean \textcolor{blue}{identifies with} the corresponding deterministic curve. \textcolor{blue}{One of the main features} of these ``dynamic'' models lays in their more accurate predictive capability with respect to the ``static'' deterministic models. Clearly, to perform real applications of these models, \textcolor{blue}{a good estimation of the involved parameters is required} (see Dey \textit{et al.}\ (2019) \cite{Dey2019} where the problem of estimation of the Bertalanffy growth model is addressed). When the likelihood function is available in closed form, one can obtain the maximum likelihood estimates directly. However, this reasoning may lead to complex systems of non-linear equations which cannot be solved analytically. For this reason, as done elsewhere (see for example Di Crescenzo \textit{et al.}\ (2022) \cite{DiCrescenzoetal2022}, Hole \textit{et al.}\ (2017) \cite{Holeetal2017}, Rom\'an-Rom\'an \textit{et al.}\ (2015) \cite{Romanetal2015} and Vera \textit{et al.}\ (2008) \cite{Vera2008}) we can adopt meta-heuristic optimization methods. In particular, the Simulated Annealing algorithm allows to obtain satisfactory estimates. In general, it is  convenient making use of gradient free algorithms, especially when the expression of the derivative of the likelihood is intricate. In any case, since the parametric space $\Theta$ (which is the set containing all the possible values of the parameters) is in principle unbounded and continuous, it is better to provide a restriction of $\Theta$ based on the knowledge of the curve. In this way, we avoid unnecessary calculations and a long running time of execution of the algorithms.
	
	The Richards growth model, as other competing models in this area, does not take into account external factors which may modify the growth rate from a certain time instant. For example, we may think about the oil production of a country. When the amount of produced oil decreases and crosses a fixed critical threshold, the government may decide to support new explorations in order to increase the quantity of \textcolor{blue}{ extracted resources}. As a further real example, we can refer to the evolution of some diseases in patients. In this case, when the monitoring parameters reach critical values, the medical team may consider to introduce a new therapy, whose effects is the stabilization of critical parameters and the reduction of the evolution of the disease. Inspired by these motivations, we would like to investigate the possible modifications of the classical Richards growth model in order to take into account the perturbations of the growth rate due to external factors that are effective from a certain moment on. Such modifications may lead to multi-sigmoidal growths, namely growths with multiple inflections. On the same line, a multi-sigmoidal version of the logistic model has been introduced in Di Crescenzo \textit{et al.}\ (2022) \cite{DiCrescenzoetal2022} by increasing the number of the involved parameters. Instead, in the present work, the proposed modified curve and the classical one possess the same number of parameters. Indeed, the modification  affects the time-dependence of one parameter.  A similar study has been conducted recently to estimate the effect of a therapy on tumor dynamics described by Gompertz diffusion processes by Albano \textit{et al.}\ (2015) \cite{Albanoetal2015}.
	The resulting modified model will be  analyzed both from a deterministic and stochastic point of view. The problem of parameters estimation will be also addressed. A strategy to obtain the maximum likelihood estimates of the parameters involved in the definition of the corresponding diffusion process will be proposed.
	
	\textcolor{blue}{
	The study introduces a time-dependent perturbation into the Richards growth model, offering an approach to take into account  external factors affecting growth rates, which is effectively demonstrated through both theoretical analysis and practical applications, in particulare the oil production in France. One of the strengths of this research lies in its comprehensive statistical analysis, employing meta-heuristic algorithms like Simulated Annealing and Ant Lion Optimizer for parameters estimation, thus ensuring robust results validated by simulation studies. The structure of the paper is organized as follows.}
	In Section \ref{Sect1}, the classical deterministic Richards growth model is introduced. Several features of the model are analyzed, such as the limit behavior with respect to the parameters, the inflection point, the approximation of the curve with a straight line near the inflection point and the first-crossing-time problem.  In Section \ref{Sect2}, we consider a modification of the classical Richards growth model, by adding a time dependent function to one of the relevant parameters. Such modification, whose consequences are visible after a certain time, called critical or switching time,  affects the growth rate which becomes, under suitable conditions, greater than the previous one. The resulting deterministic curve is then studied and a sensitivity analysis on the switching time is also considered. In Sections \ref{BDproc} and \ref{DiffProc}, the stochastic counterparts of the proposed models are introduced. In detail, we first consider two different special time-inhomogeneous birth-death processes following the line of Majee \textit{et al.}\ (2022) \cite{Majeeetal2022}. Sufficient and necessary conditions are provided in order to have a mean of the birth-death processes of modified Richards type. Moreover, in order to have a more manageable stochastic counterpart,  we define two lognormal diffusion processes whose means are of Richards and modified Richards type, respectively. Some comparisons between the two diffusion processes are also provided. In particular, one of them turns out to be very useful for the estimation of the modified model. The problem of parameters estimation is addressed in Section \ref{parest}. A procedure to estimate the modified process is described. The relevant parameters are estimated by means of maximum likelihood method since an explicit expression of the log-likelihood function is available. This function is maximized via meta-heuristic algorithms, in particular Simulated Annealing and Ant Lion Optimizer are used in Section \ref{SA}. A simulation study to validate the \textcolor{blue}{described} procedures ends the estimation study in Section \ref{Simulation}. Finally, an application to real data regarding oil production in France is studied in Section \ref{Sect9}.

	\section{The Richards growth curve}\label{Sect1}
	The Bertalanffy-Richards growth curve $x_\theta(t)$, $t\ge t_0$, is given by (cf.\ Rom\'an-Rom\'an \textit{et al.}\ (2015) \cite{Romanetal2015})
	\begin{equation}\label{BR}
		x_\theta(t)=x_0\left(\frac{\eta + k^{t_0}}{\eta+k^t}\right)^q, \qquad t\ge t_0,
	\end{equation}
	where $\theta:=(q,k,\eta)^T$ is the vector of the parameters with $q>0$, $0<k<1$ and $\eta >0$. The function \eqref{BR}  is a generalization of the logistic growth curve which can be obtained for $q=1$.
	The main difference between Eq.\ \eqref{BR} and the logistic model lays in the behavior of the curve at the inflection point. Indeed, the value at inflection point of the logistic model equals a half of the carrying capacity (namely the asymptotic value of a growth curve) whereas the one of Bertalanffy-Richards is equal to a (possibly) different fraction of the carrying capacity.
	In particular, the carrying capacity of  \eqref{BR} depends explicitly on the initial value $x_0$. Indeed, it is given by
	\begin{equation}\label{carcap}
		\mathcal K_\theta:=\lim_{t\to+\infty}x_\theta(t)=x_0\left(1+\frac{k^{t_0}}{\eta}\right)^q.
	\end{equation}
				The function \eqref{BR}, is the solution of a particular inhomogeneous Malthusian equation having a time-dependent growth rate $h_\theta(t)$, i.e.
				\begin{equation}\label{MalEq}
					\frac{\textrm{d}}{\textrm{d}t}x_\theta(t)=h_{\theta}(t) x_\theta(t), \quad t\ge t_0,\qquad x_\theta(t_0)=x_0,
				\end{equation}
				where
				\begin{equation}\label{hclas}
					h_{\theta}(t):=q\frac{k^t|\log k|}{\eta+ k^t}.
				\end{equation}
				
				It is not hard to see that the growth rate $h_\theta(t)$ is a positive and decreasing function for $t\ge t_0$.
				The function $x_\theta(t)$ is also the solution of a particular differential equation in which the time dependence of the right-hand-side is expressed only through $x_\theta(t)$, i.e.
				\begin{equation}\label{Eqdiffrip}
					\frac{\textrm{d}}{\textrm{d}t}x_\theta(t)=q \mid\log k\mid x_\theta(t)\left[1-\frac{\eta}{\eta+k^{t_0}}\left(\frac{x_\theta(t)}{x_0}\right)^{1/q}\right],\qquad t\ge t_0.
				\end{equation}
				Note that  when $\eta\to+\infty$, $q\to+\infty$ and $\left(\frac{\eta}{\eta+k^{t_0}}\right)^q\to\frac{x_0}{\mathcal K_\theta}$, Eq.\ \eqref{Eqdiffrip} becomes
				\begin{equation*}
					\frac{\textrm{d}}{\textrm{d}t}x_\theta(t)=\mid \log k \mid x_\theta(t) \log\frac{\mathcal K_\theta}{x_\theta(t)},\qquad t\ge t_0,
				\end{equation*}
				which is a particular Gompertz equation having a carrying capacity given by $\mathcal K_\theta$.
				\begin{remark}\label{rem1.1}
					Without loss of generality, it is possible to take $t_0=0$. Indeed, by setting $t':=t-t_0$, we get a model of the same type of \eqref{BR}. Precisely,
					$$
					x_\theta(t)=x_0\left(\frac{\eta +k^{t_0}}{\eta +k^t}\right)^q=x_0\left(\frac{\widehat\eta +1}{\widehat\eta +k^{t'}}\right)^q=:y_\theta(t'),\qquad\qquad t'\ge 0,
					$$
					where $\widehat\eta:=\eta/k^{t_0}$. We point out that the parameter $k$ is the same for both the models.
				\end{remark}
				In Table \ref{Tab:Table1}, we provide some limit behaviors of the proposed growth model (in agreement with the results given by Albano \textit{et al.}\ (2022) \cite{Albanoetal2022}). Note that for $\eta\to 0^+$, the Richards curve $x(t)$ converges to a Malthusian growth curve.
				\begin{table}[h]
					\caption{The growth equation, the growth function and the carrying capacity of the model \eqref{MalEq} under the specified limit conditions, where $x_\theta'(t)=\displaystyle\frac{d}{dt}x_\theta(t)$.}
					\label{Tab:Table1}
					\centering
					\begin{tabular}{lll}
						Case no.  &Limit conditions  & Growth equation   		  \\ \hline
						$1$	&$k\to0^+$         & $x_\theta'(t)=0$ \\
						$2$ &$k\to1^-$         & $x_\theta'(t)=0$                  \\
						$3$ &$\eta\to0^+$      & $x_\theta'(t)=q\mid\log k\mid x_\theta(t)$                   \\
						$4$ &$q\to 0^+$        & $x_\theta'(t)=0$                              \\
						$5$ &$q\to 1$		&$x_\theta'(t)=\frac{k^t|\log k|}{\eta+ k^t}x_\theta(t)$	\\
						$6$   &$\eta,q\to+\infty,\; \left(\frac{\eta}{\eta+k^{t_0}}\right)^q\to\frac{x_0}{\mathcal K}$  & $x_\theta'(t)=\mid \log k\mid x_\theta(t)\log\frac{\mathcal K_\theta}{x_\theta(t)}$
					\end{tabular}
					
					\vspace{0.5cm}
					
					\begin{tabular}{lll}
						Case no.    		& Growth function            &Carrying capacity  \\ \hline
						$1$                            & $x_\theta(t)=x_0$	&$x_0$ \\
						$2$                               & $x_\theta(t)=x_0$                   &$x_0$   \\
						$3$                   & $x_\theta(t)=x_0k^{-qt}$             &$+\infty$     \\
						$4$                               & $x_\theta(t)=x_0$     &$x_0$          \\
						$5$		& $x_\theta(t)=x_0\left(\frac{\eta + k^{t_0}}{\eta+k^t}\right)$ &$x_0\left(\frac{\eta+k^{t_0}}{\eta}\right)$\\
						$6$                    &  $x_\theta(t)=\mathcal K_\theta \exp[-\log(\mathcal K_\theta/x_0)e^{\log k (t-t_0)}]$   &$\mathcal K_\theta$
					\end{tabular}
				\end{table}
				
				\subsection{Inflection point and related quantities}
				The inflection point $t_I$ of the Bertalanffy-Richards curve \eqref{BR} has an explicit expression given by
				\begin{equation}\label{infl}
					t_{I}=\frac{\log(\eta/q)}{\log{k}},\qquad x_\theta(t_I)=\mathcal K_\theta \left(\frac{q}{1+q}\right)
					^q.
				\end{equation}
				See Figure \ref{fig:Figure3} for some plots of the Bertalanffy-Richards growth curve.
				\begin{figure}[h]
					\centering
					\subfigure[]{\includegraphics[scale=0.6]{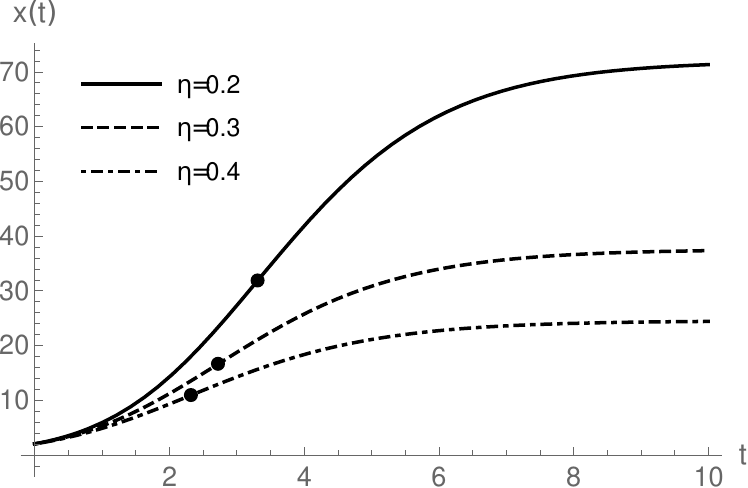}}
					\subfigure[]{\includegraphics[scale=0.6]{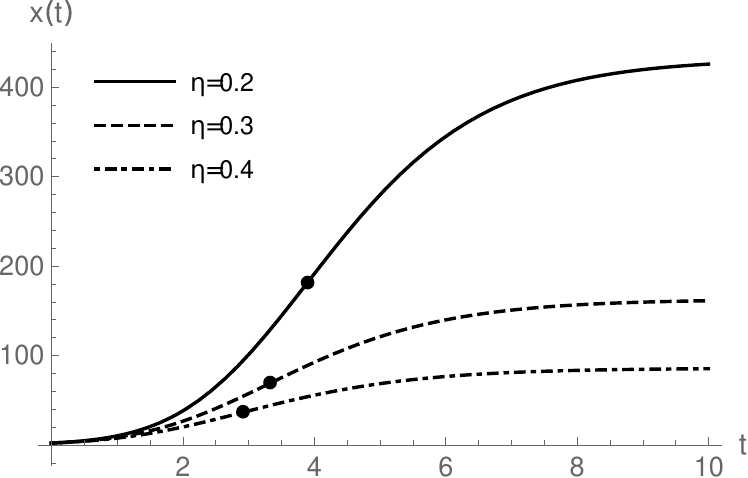}}
					\caption{The function $x_\theta(t)$  and the inflection point $t_I$ (dot) given in Eq.\ \eqref{infl} with $t_0=0$, $x_0=2$, $k=0.5$, $\eta=0.2,0.3,0.4$ and (a) $q=2$ and (b) $q=3$. }
					\label{fig:Figure3}
				\end{figure}
				
				From Eq.\ \eqref{infl}, we note that the \textcolor{blue}{ratio between $x_\theta(t_I)$ and the carrying capacity $\eqref{carcap}$} depends only on $q$, and it  is given by $\left(\frac{q}{1+q}\right)^q$. In the logistic case (obtained for $q=1$), this fraction is equal to $1/2$. Moreover, one has $t_I>t_0$ if and only if  $\eta<qk^{t_0}$.
				Now, we can provide an interpretation of the relevant parameters of the model, i.e. $\theta=(q,k,\eta)^T$. It turns out that
				\begin{enumerate}
					\item[$\bullet$] The parameter $q$ affects the ratio between the carrying capacity $\mathcal K_\theta$ of the model and the inflection point $x_\theta(t_I)$. Indeed, the ratio $\mathcal K_\theta/x_\theta(t_I)=\displaystyle\left(1+1/q\right)^q$ is increasing in $q>0$.
					\item[$\bullet$] The parameter $\eta$ represents a measure of the distance between $x_\theta(t)$ and the exponential function $x_0k^{-qt}$. Indeed, the bigger $\eta$ is, the further the function $x_\theta(t)$ is from being the exponential function $x_0k^{-qt}$ (see also Figure \ref{Fig:Figure24gennaio}). In particular, as shown in Table \ref{tab:Table1}, for $\eta\to 0^+$ the function $x_\theta(t)$ converges to $x_0k^{-qt}$.
					\item[$\bullet$] The parameter $k$ affects the value of the inflection point $t_I$. Indeed, for $\eta<qk^{t_0}$ the value of $t_I$ increases for increasing $k$.
				\end{enumerate}
				
				\begin{figure}[h]
					\centering
					\subfigure[]{\includegraphics[scale=0.6]{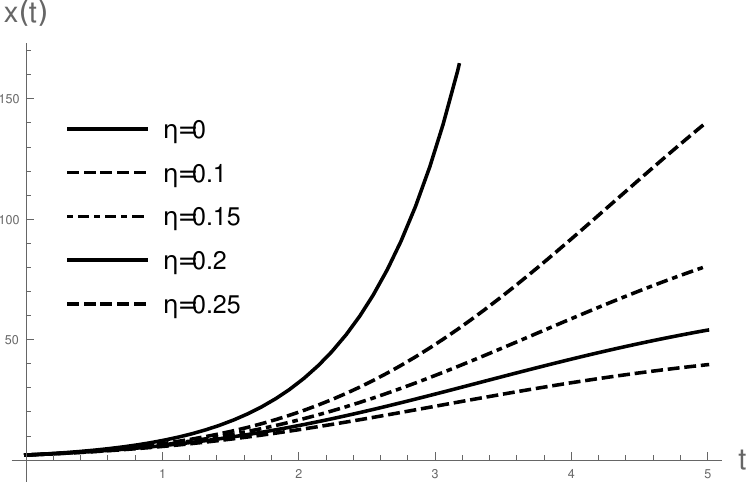}}
					\subfigure[]{\includegraphics[scale=0.6]{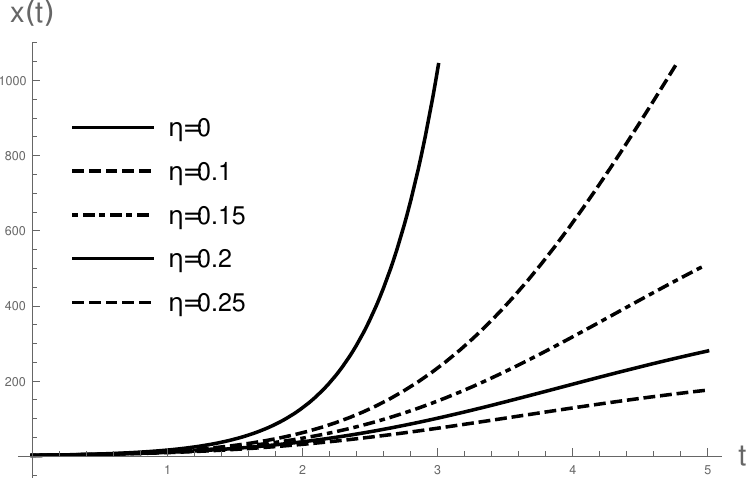}}
					\caption{The function $x_\theta(t)$ with $t_0=0$, $x_0=2$, $k=0.5$, (a) $q=2$ and (b) $q=3$ for $\eta=0,0.1,0.15,0.2,0.25$.}
					\label{Fig:Figure24gennaio}
				\end{figure}
				
				As already done in other similar researches (see, for example, Asadi \textit{et al.}\ (2020) \cite{Asadietal2020} and Di Crescenzo \textit{et al.} (2022) \cite{DiCrescenzoetal2022}), we can analyze the behavior of the growth curve $x_\theta(t)$ around the inflection point $t_I$ by using a linear approximation. With this aim, we consider the maximum specific growth rate, denoted by $\mu$, which is the slope of the line tangent to $x_\theta(t)$ in $t_I$ and the lag time $\lambda$ which is the intersection between the $x$-axis and the tangent. For the Bertalanffy-Richards curve, these quantities are given by
				$$
				\mu=(\eta+k^{t_0})^q \, \frac{x_0 \mid \log{k}\mid}{\eta^q}\left(\frac{q}{q+1}\right)^{q+1}>0,\qquad\qquad \lambda=t_I-\frac{1+1/q}{|\log{k}|}<t_I.
				$$
				Note that if $k\to 1$ or $q\to 0$, then $\mu$ tends to $0$  and thus the tangent line tends to be parallel to the $x$-axis. Indeed, in these limit cases the curve $x_\theta(t)$ degenerates into a horizontal line.
				In Figure \ref{fig:Figure4}, we provide the plot of the line tangent to the Bertalanffy-Richards curve at the inflection time instant $t_I$.
				\begin{figure}[h]
					\centering
					\subfigure[]{\includegraphics[scale=0.6]{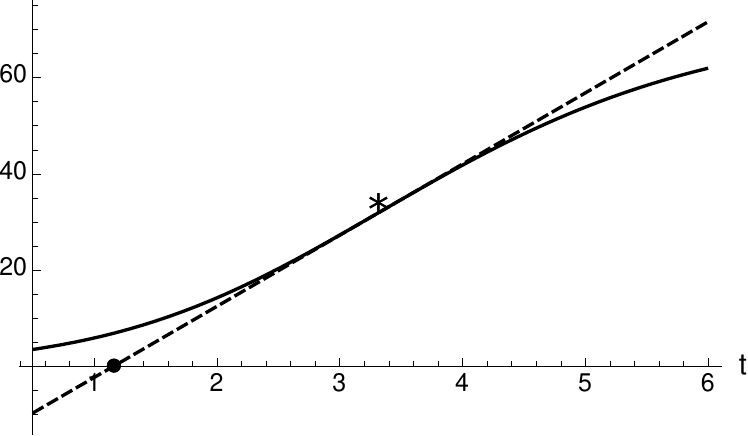}}
					\subfigure[]{\includegraphics[scale=0.6]{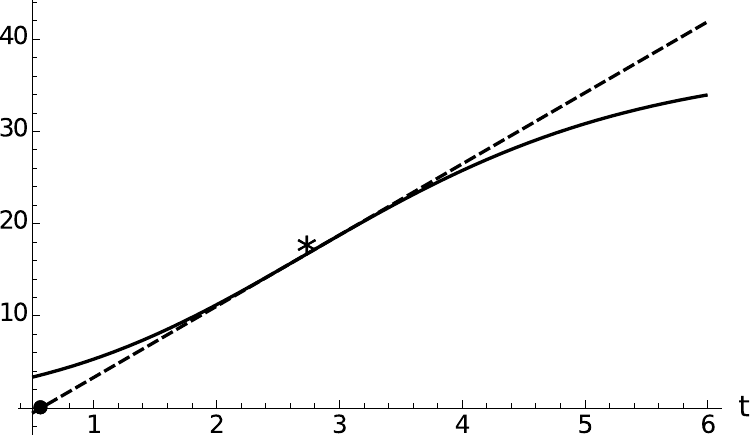}}
					\caption{The function $x_\theta(t)$ (solid line), the line tangent to the curve (dashed line), the inflection point (star point) and the lag time (dot point) in $t_I$ (dot) with $t_0=0$, $x_0=2$, $k=0.5$, $q=2$ and (a) $\eta=0.2$  and (b) $\eta=0.3$. In (a) $\mu\simeq14.79$ and $\lambda\simeq1.16$. In (b) $\mu\simeq7.71$ and $\lambda\simeq0.57$.}
					\label{fig:Figure4}
				\end{figure}
				
				\subsection{Threshold crossing problem}
				In this section, we consider the problem of first crossing time. Let $B(t)$ be a time-dependent boundary given by
				$$
				B(t):=(1+p)x_\theta(t),\qquad t\ge t_0, \qquad p>0.
				$$
				The corresponding first crossing time $\theta_t$ for a fixed time instant $t$ is then defined as follows
				$$
				\theta_t:=\min\{s\ge t_0: x_\theta(s)=B(t)\},
				$$
				and it represents the first instant in which the growth curve $x_\theta(t)$ crosses the threshold $B(t)$. Clearly, if $B(t)$ is greater than the carrying capacity, the set $\{s\ge t_0: x_\theta(s)=B(t)\}$ is empty and then $\theta_t:=+\infty$. Otherwise, $\theta_t$ is finite and it can be explicitly determined. The expression one gets is the following
				$$
				\theta_t=\frac{1}{\log k}\log \left(\frac{\eta+k^t}{(1+p)^{1/q}}-\eta\right)=t_I+\frac{\log\left(\frac{q(\eta+k^t)}{\eta(1+p)^{1/q}}-q\right)}{\log k}.
				$$
				Note that when $t=t_0$, then $B(t_0)=(1+p)x_0>x_0$ and
				$$
				\theta_{t_0}=\frac{\log\left(\frac{\eta+k^{t_0}}{(1+p)^{1/q}}-\eta\right)}{\log k}.
				$$
				Moreover, when $t=t_I$, then the boundary is
				\begin{equation}\label{Sdef}
					S:=B(t_I)=(1+p)x_\theta(t_I)
				\end{equation}
				and the corresponding first crossing time is given by
				\begin{equation}\label{tstar}
					t^*:=\theta_{t_I}=t_I+\frac{\log\left ( \frac{1+q}{(1+p)^{1/q}}-q \right )}{\log k}>t_I.
				\end{equation}
				
				\section{A modified model}\label{Sect2}
				Once the parameters of the model are set, the evolution of the growth curve follows the resulting pattern and no further modifications are considered. On the other hand, there are several real situations in which the growth rate may be modified due to the presence of external factors. Usually, the effects produced by external modifications become significant starting from a critical time. As an example, think about the case of oil production of a country: when the quantity of produced oil is lower than a fixed threshold, the goverment may decide to support new explorations aiming to increase the amount of production. \\
				For this reason we consider a modified version of the classical Bertalanffy-Richards growth curve by introducing a time-varying term into the growth rate  $h_{\theta}(t)$ of the Malthusian Eq.\ \eqref{MalEq}. Specifically, in Eq.\ \eqref{BR} we substitute the parameter $q$ with a time-dependent one defined as
				\begin{equation}\label{qtilde}	
					\widetilde q (t):=\begin{cases}
						q,\qquad &t_0\le t\le t^*\\
						q+C(t),\qquad &t>t^*,
					\end{cases}
				\end{equation}
				where $t^*\ge t_0$  and $C(t)$ is a continuous, bounded and positive function with $\displaystyle \lim_{t\to t^*}C(t)=0$.
				Hence, the growth rate $\widetilde h_{\theta}(t)$ of the new model is given by
				\begin{equation}\label{hgen}
					\widetilde h_{\theta}(t)=\frac{\widetilde q(t)}{q}h_\theta(t)=\widetilde q(t)\frac{k^t \mid \log k \mid}{\eta+k^t},\qquad t\ge t_0.
				\end{equation}
				From now on, we consider $t^*>t_I$ as the time instant defined in Eq.\ \eqref{tstar}.
				Note that the conditions verified by the function $C(t)$ ensure the continuity of the function $\widetilde q(t)$ and consequently of the function $\widetilde h_{\theta}(t)$.
				It can be also noticed that the adoption of the new term given in \eqref{qtilde} increases the growth rate, indeed
				$$
				\widetilde h_{\theta}(t)-h_{\theta}(t)=C(t)\frac{k^t \mid \log{k}\mid }{\eta+k^t}>0, \qquad t\ge t^*.
				$$ 	
				Let us denote by $\widetilde x_\theta(t)$ the corresponding growth curve, which is the solution of Eq.\ \eqref{MalEq} where the function $h_{\theta}(t)$ is replaced by $\widetilde h_{\theta}(t)$. In particular, we have that for any $t\ge t_0$	
				\textcolor{blue}{
				\begin{equation}\label{GenBR}
					\begin{aligned}
						\widetilde x_\theta(t)&=x_\theta(t)\exp\left(\int_{t^*}^{\max\{t,t^*\}} C(s)\frac{k^s\mid \log{k}\mid }{\eta+k^s}\mathrm d s\right)\\
						&=\begin{cases}
							x_\theta(t),&\; t_0\le t<t^*\\
							x_\theta(t)\exp\left(\displaystyle\int_{t^*}^{t} C(s)\frac{k^s|\log k|}{\eta+k^s}\mathrm {d}s\right), &\; t\ge t^*.
						\end{cases}
					\end{aligned}
				\end{equation}
			}
				We remark that $\widetilde x_\theta(t)\ge x_\theta(t)$ for any $t\ge t_0$.
				The carrying capacity of the new growth curve \eqref{GenBR} is given by
				\begin{equation*}
					\widetilde{ \mathcal K_\theta}=\mathcal K_\theta \exp\left(\int_{t^*}^{+\infty}C(s)\frac{k^s|\log k|}{\eta+k^s}\mathrm{d}s\right)>\mathcal K_\theta,
				\end{equation*}
				where $\mathcal K_\theta$ is given in Eq.\ \eqref{carcap}. 	
				\begin{example}
					Figure \ref{fig:Figure1}(a) provides some plots of $C(t)$, where
					\begin{equation}\label{Cta}
						C(t)=\left(\frac{1}{\eta+k^t}\right)^m-\left(\frac{1}{\eta+k^{t^*}}\right)^m, \qquad t\ge t^*,
					\end{equation}
					with $m>0$. In this case the function $C(t)$ has a downward concativity, for any $t\ge t^*$.
					\begin{figure}[h]
						\centering
						\subfigure[]{\includegraphics[scale=0.55]{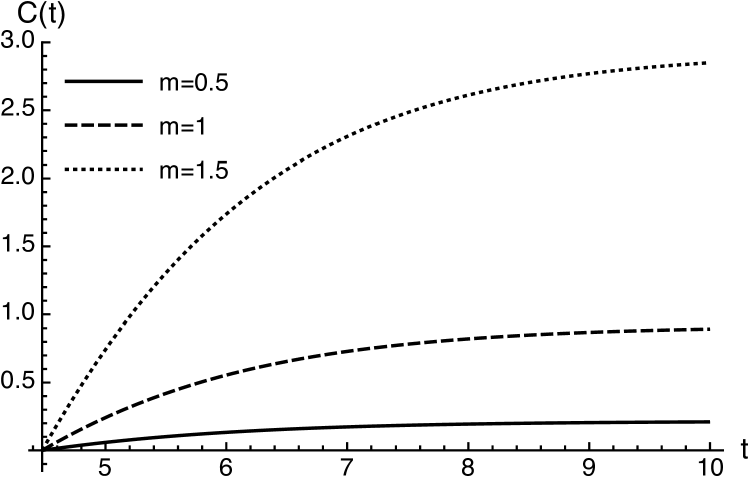}}\qquad
						\subfigure[]{\includegraphics[scale=0.55]{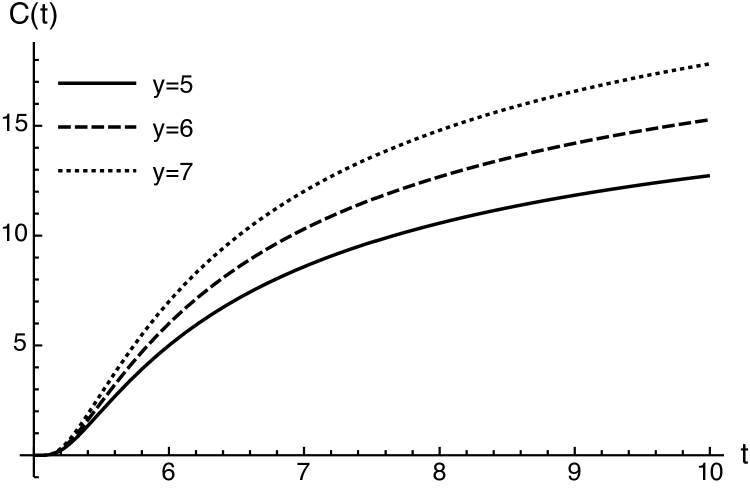}}
						\caption{(a) The function $C(t)$  given in Eq.\ \eqref{Cta} with  $m=0.5,1,1.5$, $p=0.3$, $t_0=0$, $x_0=2$. (b) The function $C(t)$ given in Eq.\ \eqref{Ctb} with $y=5,6,7$, $\alpha=1$, $\beta=0.75$, $p=0.8$. In both cases, we have  $k=0.5$, $\eta=0.2$, $q=2$. }
						\label{fig:Figure1}
					\end{figure}
					In Figure \ref{fig:Figure2}(a) the behavior of the corresponding modified model $\widetilde x_\theta(t)$ is shown.\\
					\begin{figure}[h]
						\centering
						\subfigure[]{\includegraphics[scale=0.55]{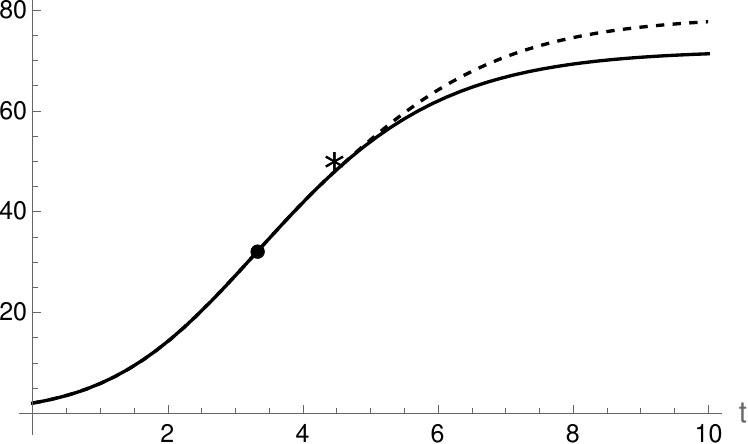}}\quad
						\subfigure[]{\includegraphics[scale=0.55]{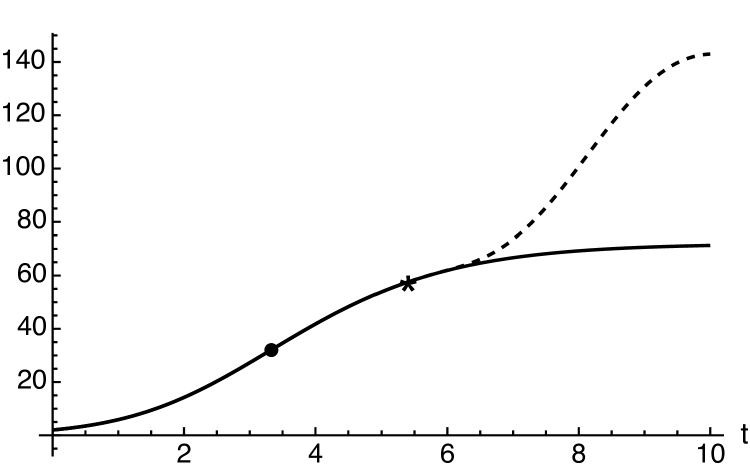}}\\
						\caption{The Bertalanffy-Richards curve (solid line) and the modified curve (dashed line) (a) for $C(t)$ given in Eq.\ \eqref{Cta} with $p=0.3$, $m=1$ and (b) for $C(t)$ given in Eq.\ \eqref{Ctb} with $p=0.8$, $y=2$, $\alpha=4$, $\beta=0.5$. In both cases $\eta=0.2$, $t_0=0$, $x_0=2$, $k=0.5$, $q=2$. The dots and the star points represent the inflection point at $t_I$ and the point at $t^*$, respectively.}
						\label{fig:Figure2}
					\end{figure}
					A different case is shown in Figure \ref{fig:Figure1}(b) where
					\begin{equation}\label{Ctb}
						C(t)=y\exp\left(\frac{\alpha}{\beta}\left[1-(t-t^*)^{-\beta}\right]\right),\qquad t>t^*,
					\end{equation}
					with $y,\alpha,\beta>0$. In this case, the function $C(t)$ has a sigmoidal behavior.
					Figure \ref{fig:Figure2}(b) illustrates the behavior of the corresponding modified model $\widetilde x_\theta(t)$, which is multi-sigmoidal with multiple inflections.
				\end{example}
				In general, since $C(t)$ is a bounded function, one has $\widetilde{\mathcal K_\theta}<+\infty$. The difference between the original curve $x_\theta(t)$ and the modified one $\widetilde x_\theta(t)$ becomes more and more relevant after the time instant $t^*$. Indeed, for $t_0<t\le t^*$ it results $\widetilde x_\theta(t)=x_\theta(t)$, whereas for $t>t^*$
				$$
				\frac{\textrm{d}}{\textrm{d}t}[\widetilde x_\theta(t)-x_\theta(t)]=x_\theta(t)\left[\widetilde h_{\theta}(t)\exp\left(\displaystyle \int_{t^*}^t C(s)\frac{k^s|\log k|}{\eta+k^s}\mathrm{d}s\right)- h_\theta (t)\right]>0.
				$$
				\begin{remark}
					Similarly as Remark \ref{rem1.1}, also for the modified model \eqref{GenBR} one can take $t_0=0$ without loss of generality. Indeed, by setting $\widehat C(t'):=C(t)$ with $t':=t-t_0$, $\widehat {t}^*:=t^*-t_0$ and $\widehat \eta:=\eta/k^{t_0}$ we have
					\begin{equation*}
						\begin{aligned}
							\widetilde x_\theta(t)&=
							x_0\left(\frac{\eta +k^{t_0}}{\eta+k^{t}}\right)^q \exp\left(\int_{t^*}^t C(s)\frac{k^s |\log k|}{\eta+k^s}\mathrm{d}s\right)\\
							&=x_0\left(\frac{\widehat\eta +1}{\widehat\eta+k^{t'}}\right)^q \exp\left(\int_{\widehat{t}^*}^{t'} \widehat C(s')\frac{k^{s'} |\log k|}{\widehat \eta+k^{s'}}\mathrm{d}s'\right)=:\widehat{y}_\theta (t'),\qquad\qquad t'\ge 0.\\
						\end{aligned}
					\end{equation*}
					It is easy to note that the curve $\widehat y_\theta(t')$ is of the same type of \eqref{GenBR}.
				\end{remark}	
				
				\subsection{Sensitivity analysis}\label{sensan}
				In this section we analyze the consequences of a perturbation on the time instant $t^*$ involved in the definition of the modified growth curve $\widetilde x_\theta(t)$.  Hereafter, we will denote the modified curve by $\widetilde x_\theta^{t^*}$ and the modified function $C_{t^*}(t)$ to point out their dependence on the time instant $t^*$. To perform the sentitivity analysis, we expand $\widetilde x_\theta^{t^*+\varepsilon}$ in a Taylor series centered in $t^*$ with $\varepsilon>0$. Therefore, one has
				\begin{equation*}
					\widetilde x_\theta^{t^*+\varepsilon}-\widetilde x_\theta^{t^*}= \varepsilon \cdot \frac{\partial}{\partial t^*} \widetilde x_\theta^{t^*}+o(\varepsilon)
				\end{equation*}
				with $\displaystyle\lim_{\varepsilon\to 0^+} \frac{o(\varepsilon)}{\varepsilon}=0$. Since
				\begin{equation*}
					\frac{\partial}{\partial t^*} \widetilde x_\theta^{t^*}(t)=\widetilde x_\theta^{t^*}(t)\cdot \int_{t^*}^t\left(\frac{\partial}{\partial t^*}C_{t^*}(s)\right)\frac{k^s \mid \log k \mid }{\eta+k^s}\mathrm d s, \qquad t>t^*,
				\end{equation*}
				we have that
				\begin{equation*}
					\sgn\left(\widetilde x_\theta^{t^*+\varepsilon}-\widetilde x_\theta^{t^*}\right)=\sgn\left(\int_{t^*}^t\left(\frac{\partial}{\partial t^*}C_{t^*}(s)\right)\frac{k^s\mid \log k\mid }{\eta+k^s}\textrm{d}s\right),\qquad  t>t^*.
				\end{equation*}

				Whereas, $\widetilde x_\theta^{t^*+\varepsilon}-\widetilde x_\theta^{t^*}=0$ for $t\le t^*$.
				As an example, in Figure \ref{fig:Figure5} we plot the modified curve $\widetilde x_\theta(t)$ and the effect of the perturbation for $\varepsilon=0.5$ and for the function $C(t)$ given in Eq.\ \eqref{Cta}. The curve $\widetilde x_\theta^{t^*+\varepsilon}$ goes down as the parameter $\varepsilon$ increases since the derivative of $C(t)$ with respect to $t^*$ is negative.
				\begin{figure}[h]
					\centering
					\subfigure[]{\includegraphics[scale=0.55]{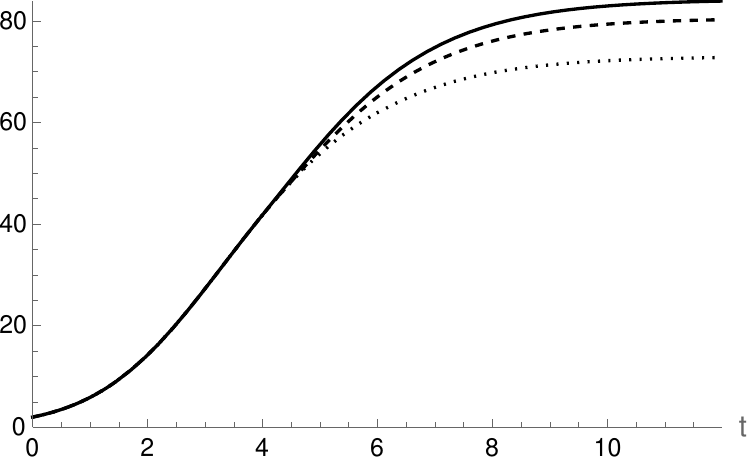}}\quad
					\subfigure[]{\includegraphics[scale=0.55]{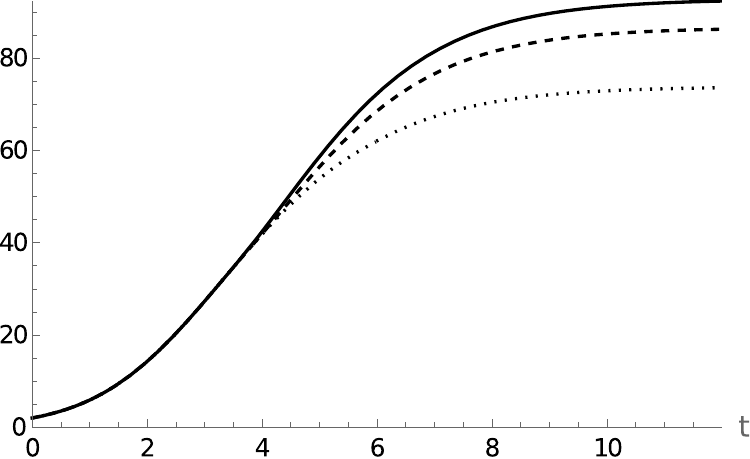}}\\
					\caption{The modified Bertalanffy-Richards curve $\widetilde x_\theta^{t^*}$ (solid line) and the perturbed ones  $\widetilde x_\theta^{t^*+\varepsilon}$ (dashed and dotted lines) for  $x_0=2$, $\eta=0.2$, $k=0.5$, $t_0=0$, $q=2$, $m=1$, (a) $p=0.3$ and (b) $p=0.1$ with $\varepsilon=0.3$ (dashed) and $\varepsilon=2$ (dotted).}
					\label{fig:Figure5}
				\end{figure}
				\section{A special linear birth-death process}\label{BDproc}
				In this section, we introduce an evolutionary model based on the birth-death process introduced in Majee \textit{et al.}\ (2022) \cite{Majeeetal2022}. In more detail, we suppose that the population size can be described by an inhomogeneous linear birth-death process $\{X(t);\,t\ge 0\}$ with state-space $\mathbb N_0=\{0,1,2,\dots\}$ where the state $0$ is an absorbing endpoint. The existence of an absorbing endpoint reflects the situations in which the extinction of the population may occur. Moreover, we consider a positive and integrable function on $(0,t)$ for any $t>0$, denoted by $\lambda(t)$ ($\mu(t)$), which represents the individual birth (death) rate at time $t$. Since in real populations there are some individuals who have no reproduction power, we denote by $\rho(t)\in[0,1]$ this portion evaluated at time $t$. We assume that $\rho(t)$ is time-varying, since the reproduction power of individuals may change over time, due to diseases or other natural reasons. Therefore, the individual transition rates of this special birth-death process are given by
				\begin{equation*}
					\lambda(t):=(1-\rho(t))\lambda,\qquad \mu(t):=\mu,
				\end{equation*}
				where $\lambda,\mu>0$ and $\rho(t)\in[0,1]$ is an integrable function over any interval $(0,\bar t)$ for $\bar t>0$.
				\par
				\noindent Hence, the transition rates are
				\begin{equation}\label{transrates}
					\begin{aligned}
						&\lambda_n(t):=\lim_{h\to 0^+} \frac{1}{h}\mathsf P[X(t+h)=n+1\mid X(t)=n]=n\lambda(t)=n(1-\rho(t))\lambda, \qquad n\in\mathbb N_0,\\
						&\mu_n(t):=\lim_{h\to 0^+} \frac{1}{h}\mathsf P[X(t+h)=n-1\mid X(t)=n]=n\mu(t)=n\mu, \qquad n\in\mathbb N.
					\end{aligned}
				\end{equation}
				We denote by
				\begin{equation*}
					{P}_{yx}(t)=\mathsf P[X(t)=x \mid X(0)=y], \qquad t\ge 0,
				\end{equation*}
				the probability that the birth-death process $X(t)$ is in the state $x$ at the time $t$, conditional on the initial state $X(0)=y\in \mathbb N$.
				As shown in Tan (1986) \cite{Tan1986}, the probability generating function of $X(t)$ has the following expression
				\begin{equation*}
					G(z,t)=\left\{1-(z-1)[(z-1)\phi(t)-\psi(t)]^{-1}\right\}^y, \qquad 0<z<1,\;t\ge 0,
				\end{equation*}
				where
				\begin{equation*}\begin{aligned}
						&\psi(t)=\exp\left(t(\mu-\lambda)+\lambda\int_0^t\rho(\tau)\textrm{d}\tau\right),\\
						&\phi(t)=\lambda\int_0^t (1-\rho(\tau))\exp\left(\tau(\mu-\lambda)+\lambda\int_0^\tau\rho(s)\textrm{d}s\right)\textrm{d}\tau,\quad t\ge 0.
					\end{aligned}
				\end{equation*}
				In the following proposition, we provide a sufficient and necessary condition so that the birth-death process $X(t)$ has a modified Bertalanffy-Richards conditional mean.
				\begin{proposition}\label{mean-BDproc}
					The linear birth-death process with birth and death rates given by
					$$
					\lambda_n(t)=n\lambda(t),\qquad \mu_n=n\mu(t),\qquad t\ge 0,
					$$
					with $\lambda(t)$ and $\mu(t)$ defined in Eqs.\ \eqref{transrates},
					has conditional mean
					\begin{equation*}
						E_y(t):=\mathsf E [X(t)\mid X(0)=y]=y\left(\frac{\eta+1}{\eta+k^t}\right)^q\exp\left(\int_{t^*}^{\max\{t,t^*\}} C(s)\frac{k^s|\log{k}|}{\eta+k^s}\textrm{d}s\right),\qquad t\ge 0,
					\end{equation*}
					if, and only if,
					\begin{equation}\label{rho}
						1-\rho(t)=\frac{\mu+\widetilde h_\theta(t)}{\lambda},\qquad t\ge 0,
					\end{equation}
					with $\lambda-\mu=-q \log k>0$ and $\widetilde h_\theta(t)$ given in Eq.\ \eqref{hgen}.
				\end{proposition}
				\begin{proof}
					The conditional mean $E_y(t)$ satisfies the following differential equation
					\begin{equation*}
						\frac{\textrm{d}}{\textrm{d}t}E_y(t)=(\lambda(t)-\mu(t))E_y(t),\qquad t\ge 0.
					\end{equation*}
					On the other hand, the modified Bertalanffy-Richards function $\widetilde x_\theta(t)$ verifies the differential equation (5) with time-dependent growth rate $\widetilde h_\theta(t)$. So, the function $\rho(t)$ must be chosen as in  Eq.\ \eqref{rho}.
				\end{proof}
				From now on, we will assume that the condition \eqref{rho} holds. Consequently,
				\begin{equation}\label{psiphipart}\begin{aligned}
						&\psi(t)=\left(\frac{\eta+k^t}{\eta +1}\right)^q\exp\left(\int_{t^*}^{\max\{t,t^*\}}C(s)\frac{k^s\log{k}}{\eta+k^s}\textrm{d}s\right),\\
						&\phi(t)=-\exp\left(-\int_0^t \widetilde h_{\theta}(u)du\right)+\mu\int_0^t\exp\left(-\int_0^s\widetilde h_{\theta}(u)\textrm{d}u\right)\textrm{d}s,\quad t\ge 0.
					\end{aligned}
				\end{equation}
				As shown in Tan (1986) \cite{Tan1986}, the transition probabilities and the conditional variance of the process can be expressed as follows, for any $t\ge 0$,
				\begin{equation*}
					\begin{aligned}
						&P_{y0}(t)=\left(1-\frac{1}{\psi(t)+\phi(t)}\right)^y,\\
						&P_{yx}(t)=\left(\frac{\phi(t)}{\psi(t)+\phi(t)}\right)^x \sum_{i=0}^{\min(x,y)}\binom{y}{i}\binom{y+x-i-1}{y-1}\left(\phi(t)^{-1}-1\right)^i\left(1-\frac{1}{\psi(t)+\phi(t)}\right)^{y-i}
					\end{aligned}
				\end{equation*}
				and
				\begin{equation*}
					Var_y(t):=\mathsf{Var}[X(t)\mid X(0)=y]=y\frac{\psi(t)+2\phi(t)-1}{\psi^2(t)}.
				\end{equation*}
				Some plots of the transition probabilities are provided in Figure \ref{fig:FiguraProb} for various choices of the parameters. In these cases, the probability $P_{y0}(t)$ is increasing with respect to $t$, whereas the probability $P_{yn}(t)$ is decreasing with respect  both to $t$ and $n$.
				\begin{figure}[h]
					\centering
					\subfigure[]{\includegraphics[scale=0.56]{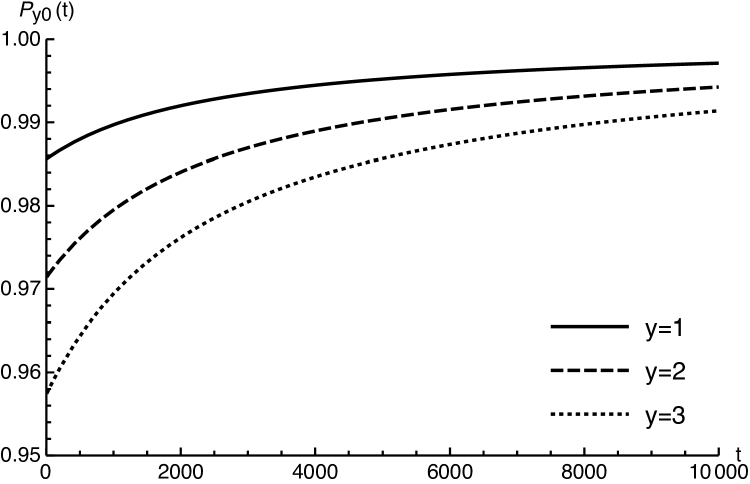}}\qquad
					\subfigure[]{\includegraphics[scale=0.56]{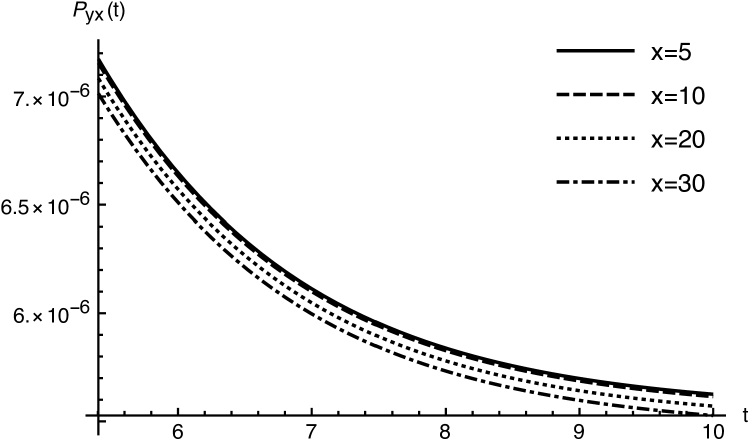}}\qquad
					\subfigure[]{\includegraphics[scale=0.56]{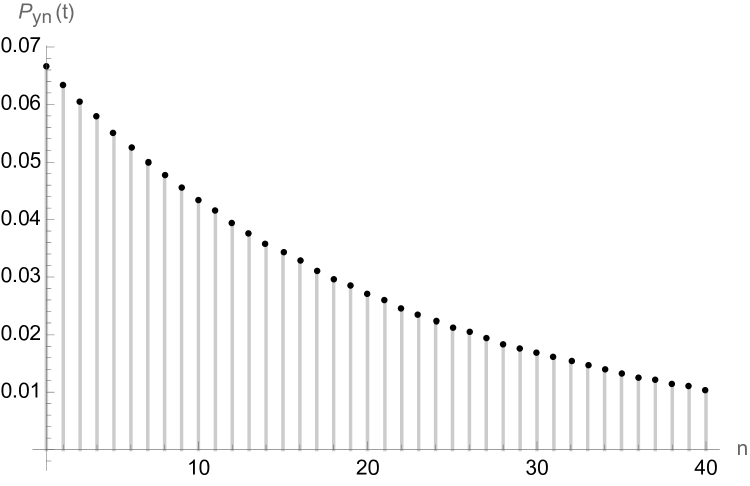}}
					\caption{The transition probabilities (a) $P_{y0}(t)$, (b) $P_{yx}(t)$ as a function of $t$  and (c) $P_{yn}(t)$ as a function of $n$ for $C(t)$ defined in Eq.\ \eqref{Cta}, with $\eta=0.2$, $m=1$, $k=0.5$, $q=2$, $p=0.8$. In (a), one has $\mu=1$, $y=1,2,3$. In (b), one has $\mu=1$, $y=1$ and $x=5,10,20,30$. In (c), one has $\mu=0.01$, $t=6$ and $y=1$.}
					\label{fig:FiguraProb}
				\end{figure}
				It is worth to notice that, since $\widetilde\psi:=\lim_{t\to+\infty}\psi(t)<+\infty$ and $\widetilde\phi:=\lim_{t\to +\infty}\phi(t)=+\infty$, the probability of an ultimate extinction is equal to $1$, indeed
				\begin{equation*}
					\pi_{y0}:=\lim_{t\to+\infty}P_{y0}(t)=\left(1-\frac{1}{\widetilde\psi+\widetilde\phi }\right)^y=1.
				\end{equation*}
				See Figure  \ref{fig:FiguraVar} for some plots of the conditional variance. Note that since the conditional variance is bounded, the conditional mean is a significant statistic for the process.
				\begin{figure}[h]
					\centering
					\subfigure[]{\includegraphics[scale=0.6]{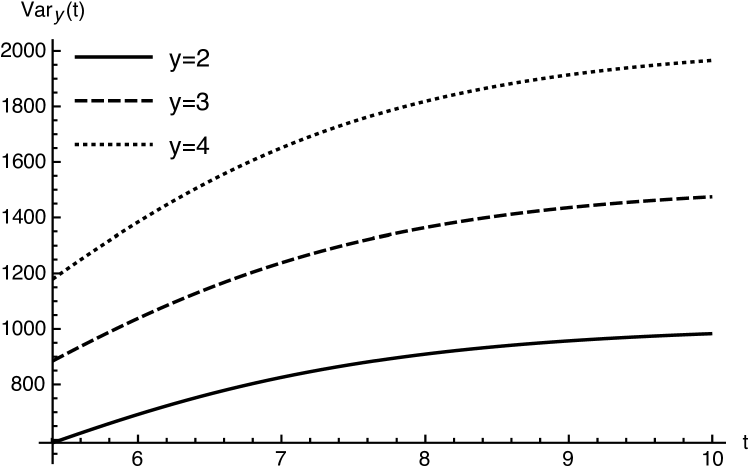}}
					\caption{The conditional variance $Var_y(t)$ as a function of $t$ with $\eta=0.2$, $m=1$, $k=0.5$, $q=2$, $\mu=0.01$, $p=0.8$ and $y=2,3,4$.}
					\label{fig:FiguraVar}
				\end{figure}
				\subsection{A special time-inhomogeneous linear birth process}\label{BDprocesses}
				In Section \ref{BDproc} we considered a birth-death process with time-dependent transition rates and we investigated the conditions under which its mean is of modified Bertalanffy-Richards type. Since the growth curve $\widetilde x_\theta(t) $ may be significantly different from the sample paths of the birth-death process (because of the presence of an absorbing endpoint), we now propose a stochastic process $X(t)$ more suitable to describe a growth phenomenon by removing the possibility of having downward jumps (deaths). Under this assumption, the individual transition rates of $X(t)$ are given by
				\begin{equation}\label{transrates1}
					\lambda(t):=(1-\rho(t))\lambda, \qquad \mu(t):=0,
				\end{equation}
				where $\lambda>0$ and $0\le\rho(t)\le 1$ is an integrable function over the interval $(0,t)$. Hence, the state-space of the process is given by $\mathcal S:=\{y,y+1,\dots\}$ being $\mathsf P[X(0)=y]=1$.
				The result of Proposition \ref{mean-BDproc} can be updated to this special case as shown in the following
				\begin{proposition}
					The linear birth process with birth rate
					\begin{equation}\label{lambdan}
						\lambda_n=n\lambda(t),\qquad t\ge 0
					\end{equation}
					with $\lambda(t)$ given in Eq.\ \eqref{transrates1},
					has conditional mean
					\begin{equation}\label{media-Bproc}
						E_y(t)=\mathsf{E}[X(t)\mid X(0)=y]=y\left(\frac{\eta+1}{\eta+k^t}\right)^q\exp\left(\int_{t^*}^{\max\{t,t^*\}} C(s)\frac{k^s|\log{k}|}{\eta+k^s}ds\right),\qquad t\ge 0,
					\end{equation}
					if, and only if,
					\begin{equation*}
						1-\rho(t)=\frac{\widetilde h_\theta(t)}{\lambda},\qquad t\ge 0,
					\end{equation*}
					with $\lambda=-q \log k>0$ and $\widetilde h_\theta(t) $ given in Eq.\ \eqref{hgen}.
				\end{proposition}
				In this case, considering the birth rate given in Eq.\ \eqref{lambdan}, the transition probabilities can be expressed as
				\begin{equation*}
					P_{yx}(t)=\mathsf{P}\left[X(t)=x\mid X(0)=y\right]=\binom{x-1}{y-1}e^{-y\Lambda(t)}\left(1-e^{-\Lambda(t)}\right)^{x-y},\qquad x\in\mathcal S,
				\end{equation*}
				where
				\begin{equation*}
					\Lambda(t)=\int_0^t\lambda(1-\rho(s))\textrm{d}s=-q\log\left(\frac{\eta+k^t}{\eta+1}\right)-\int_{t^*}^{\max\{t,t^*\}}C(s)\frac{k^s\log{k}}{\eta+k^s}\textrm{d}s,\quad t\ge 0.
				\end{equation*}
				In Figure \ref{fig:FiguraProbBproc} some plots of the transitions probabilities are provided for some choices of the parameters. In these cases, the probability $P_{yx}(t)$ is decreasing with respect to $t$ and increasing with respect to $x$.
				\begin{figure}[h]
					\centering		
					\subfigure[]{\includegraphics[scale=0.56]{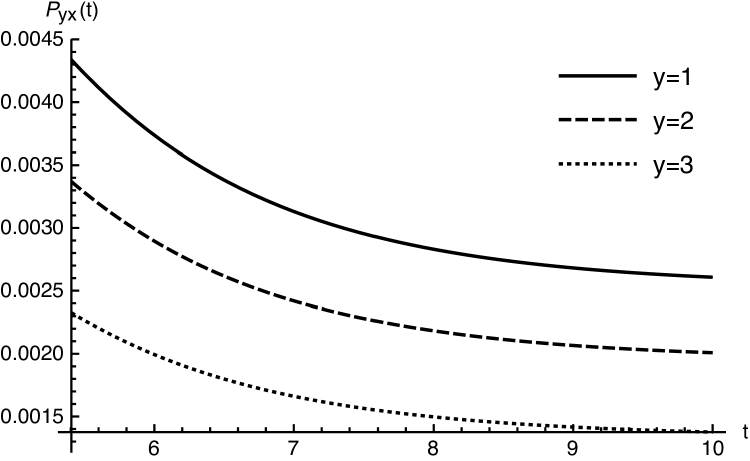}}\qquad
					\subfigure[]{\includegraphics[scale=0.56]{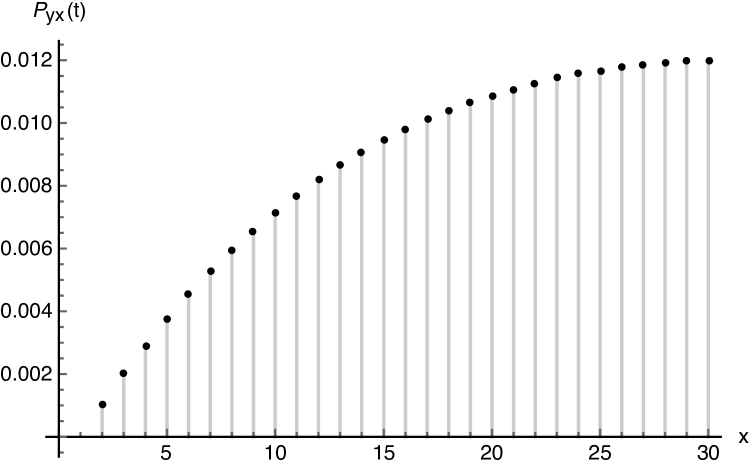}}
					\caption{The transition probabilities (a) $P_{yx}(t)$ as a function of $t$  for $x=3,4,5$ and (b) $P_{yx}(t)$ as a function of $x$ with $t=6$. In both cases, we have $C(t)$ defined in Eq.\ \eqref{Cta}, $\eta=0.2$, $m=1$, $k=0.5$, $q=2$, $p=0.8$, $y=2$. }
					\label{fig:FiguraProbBproc}
				\end{figure}
				The conditional mean of the process is given by Eq.\ \eqref{media-Bproc}, whereas, the conditional variance is
				$$
				Var_y(t)=y\, \frac{1-\psi(t)}{\psi(t)^2},\qquad t\ge0
				$$
				with $\psi(t)$ given in the first of Eqs.\ \eqref{psiphipart}.
				Let us now determine some indexes of dispersion of the process which may be useful in certain applied contexts. In particular, the Fano factor is given by
				\begin{equation}\label{fano-factor}
					D(t):=\frac{Var_y(t)}{E_y(t)}=\frac{1}{\psi(t)}-1,\qquad t\ge 0.
				\end{equation}
				Note that $D(t)$ is an increasing function with respect to $t$. From Eq.\ \eqref{fano-factor}, it is possible to note that
				\begin{itemize}
					\item[-] the birth process $X(t)$ is underdispersed for $t>\widetilde t$ being $\widetilde t$ the solution of the equation $\psi(t)=1/2$ which has a solution when $ \widetilde{\mathcal K_\theta}> 2y$;
					\item[-] the birth process $X(t)$ is overdispersed for $t<\widetilde t$.
				\end{itemize}
				Similarly, one can obtain the explicit expression for the coefficient of variation $\sigma_y(t)$ which is
				$$
				\sigma_y(t):=\frac{\sqrt{Var_y(t)}}{E_y(t)}=\sqrt{\frac{1-\psi(t)}{y}}, \qquad t\ge 0.
				$$
				We remark that $\sigma_y(t)$ is increasing in $t$ and decreasing in $y$ and the following limits hold
				\begin{equation*}
					\begin{aligned}
						\lim_{y\to 0^+}\sigma_y(t)=+\infty, \qquad \lim_{y\to\widetilde{\mathcal K_\theta}} \sigma_y(t)=0.
					\end{aligned}
				\end{equation*}
				Figure \ref{fig:FiguraVarBProcess} provide some plots of the conditional variance $Var_y(t)$, of Fano factor $D_y(t)$ and of the coefficient of variation $\sigma_y(t)$. All quantities represented are increasing and bounded.  Note that all results presented in this section are in agreement with the birth process studied in Section 4 of Di Crescenzo and Paraggio (2019) \cite{DiCrescenzoParaggio2019} related to a logistic growth curve.\\
				\begin{figure}[h]
					\centering
					\subfigure[]{\includegraphics[scale=0.56]{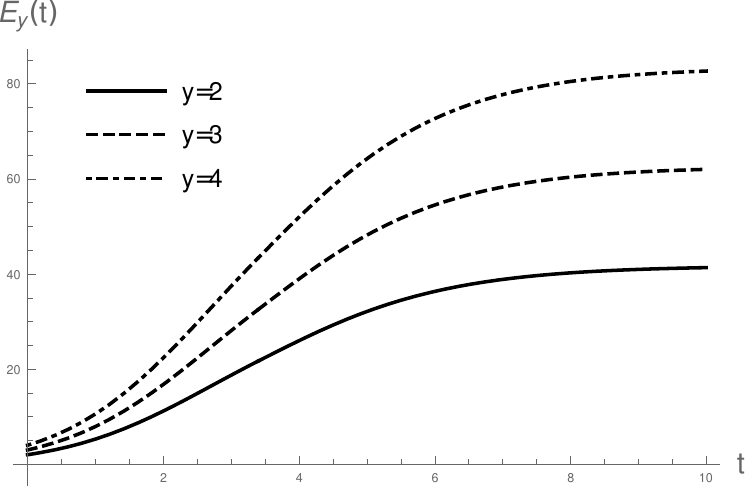}}\qquad
					\subfigure[]{\includegraphics[scale=0.56]{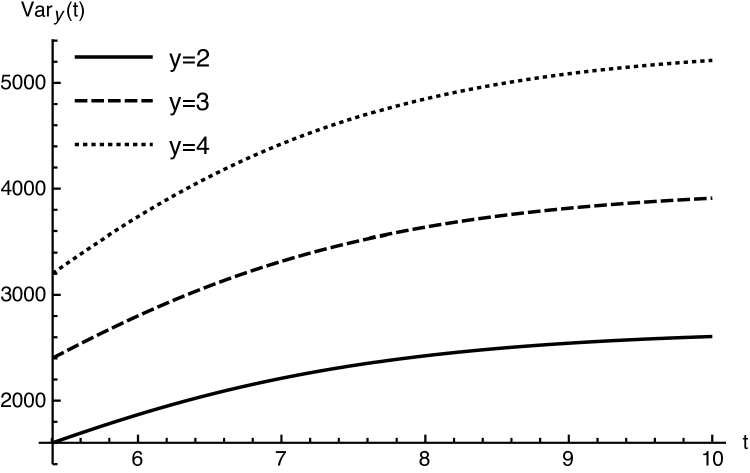}}\qquad
					\subfigure[]{\includegraphics[scale=0.56]{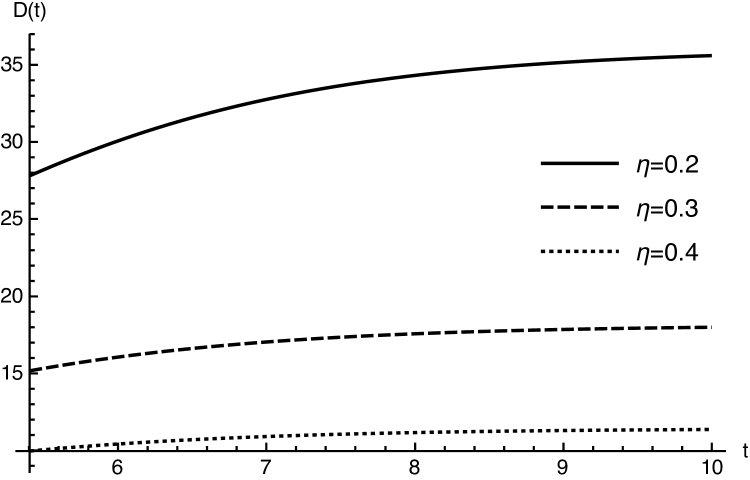}}\qquad
					\subfigure[]{\includegraphics[scale=0.56]{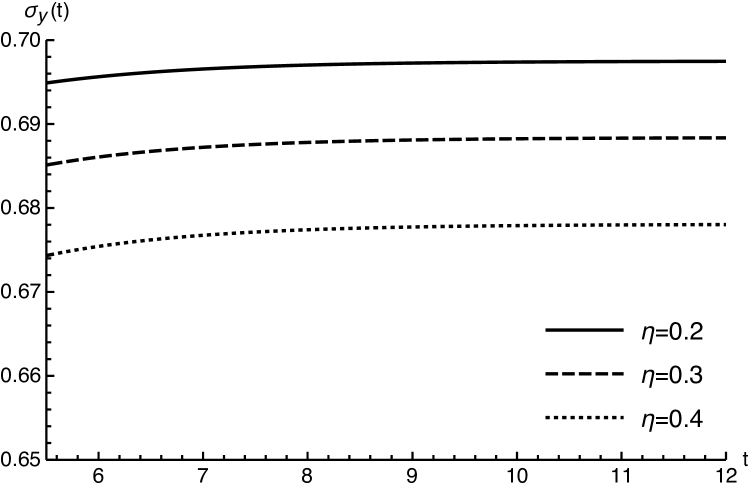}}\qquad
					\caption{(a) We consider the conditional mean $E_y(t)$ and (b) the conditional variance $Var_y(t)$ as a function of $t$ with $\eta=0.2$, $m=1$, $k=0.5$, $q=2$, $p=0.8$, $y=2,3,4$, (b) the Fano factor $D(t)$ and (c) the coefficient of variation $\sigma_y(t)$ as a function of $t$ with  $m=1$, $y=2$, $k=0.5$, $q=2$, $p=0.8$, $\eta=0.2,0.3,0.4$.}
					\label{fig:FiguraVarBProcess}
				\end{figure}
				The following time limits hold
				$$
				\lim_{t\to 0}\Lambda(t)=0,\qquad \lim_{t\to+\infty} \Lambda(t)=\log\frac{\widetilde{\mathcal K_\theta}}{y},
				$$
				$$
				\lim_{t\to 0} D(t)=0,\qquad  \lim_{t\to+\infty}D(t)=\frac{\widetilde{\mathcal K_\theta}}{y}-1,
				$$
				$$
				\lim_{t\to 0}\sigma_y(t)=0,\qquad \lim_{t\to+\infty}\sigma_y(t)=\sqrt{\frac{\widetilde{\mathcal K_\theta}-y}{\widetilde{\mathcal K_\theta} y}}.
				$$
				In order to obtain a more manageable stochastic counterpart, in Section \ref{DiffProc} we consider a particular lognormal diffusion process having the same interesting feature of the BD processes introduced so far:  its mean is of type $\widetilde x_\theta(t)$.
				\section{The corresponding diffusion processes}\label{DiffProc}
				We now consider two lognomal diffusion processes  $\{X(t); t\in I\}$ and $\{\widetilde X(t); t\in I\}$, with $I=[t_0,+\infty)$, having a mean which corresponds the former to the curve \eqref{BR} and the latter to the curve \eqref{GenBR}. The considered non-homogeneous lognormal diffusion processes can be regarded as diffusive approximations of suitable birth-death processes having quadratic rates, as shown in Section 5.2 of Di Crescenzo \textit{et al.}\ (2021) \cite{DiCrescenzoetal2021}.
				The stochastic differential equation (SDE) related to the classical model \eqref{BR} is given by
				\begin{equation}\label{SDE1}
					\textrm{d}X(t)=h_\theta(t)X(t)\textrm{d}t+\sigma X(t)\textrm{d}W(t), \qquad X(t_0)=X_0
				\end{equation}
				where $W(t)$ denotes a Wiener process independent on the initial condition $X_0$, $h_\theta(t)$ is defined in Eq.\ \eqref{hclas} and $\sigma>0$.
				\par
				Similarly, the diffusion process $\{\widetilde X(t); t\in I\}$ modelling the modified curve is the solution of the following SDE
				\begin{equation}\label{SDE2}
					\textrm{d}\widetilde X(t)=\widetilde h_\theta(t)\widetilde X(t)\textrm{d}t+\sigma \widetilde X(t)\textrm{d}W(t), \qquad \widetilde X(t_0)=\widetilde X_0
				\end{equation}
				where $\widetilde h_\theta(t)$ is defined in Eq.\ \eqref{hgen}. It is worth to remark that both Eq.\ \eqref{SDE1} and Eq.\ \eqref{SDE2} are obtained by adding to the corresponding Malthusian equations (cf. Eq.\ \eqref{MalEq}) a multiplicative noise term. \textcolor{blue}{The solutions} of the SDEs \eqref{SDE1} and \eqref{SDE2} can be easily determined by means of It\^o's formula with the variable transformation $f(\widetilde Y):=\log(\widetilde Y)$, with $\widetilde Y\in\{X(t),\widetilde X(t)\}$. Indeed, they are given by
				\begin{equation}\label{difproc1}
					X(t)=X_0 \exp\left(H_\xi(t_0,t)+\sigma \left(W(t)-W(t_0)\right)\right), \qquad t\ge t_0,
				\end{equation}
				\begin{equation}\label{difproc}
					\widetilde X(t)=\widetilde X_0 \exp\left(\widetilde H_\xi(t_0,t)+\sigma \left(W(t)-W(t_0)\right)\right), \qquad t\ge t_0,
				\end{equation}
				where, for $t>s$
				\begin{equation}\label{Htilde1}
					H_\xi(s,t)=q \log \frac{k^s+\eta}{k^t+\eta}-\frac{\sigma^2}{2}(t-s),
				\end{equation}
				\begin{equation}\label{Htilde}
					\widetilde H_\xi(s,t)=H_\xi(s,t)+\int_{\max\{s,t^*\}}^{\max\{t,t^*\}} C(u)\frac{k^u \mid \log{k}\mid}{\eta+k^u}  \mathrm d u,
				\end{equation}
				with $\xi:=(\theta^T,\sigma)^T=(q,k,\eta,\sigma)^T$.
				The  processes \eqref{difproc1} and \eqref{difproc} are lognormal diffusion processes with state-space $(0,+\infty)$
				and  infinitesimal moments
				\begin{equation*}
					\begin{aligned}
						&A_1(x,t)= h_\theta(t)x, \qquad \qquad A_2(x)=\sigma^2x^2\\
						&\widetilde A_1(x, t)=\widetilde h_\theta(t)x, \qquad \qquad \widetilde A_2(x)=\sigma^2x^2,
					\end{aligned}
				\end{equation*}
				respectively.
				In Figure \ref{fig:Figure6}, some simulated sample paths of the process $\widetilde X(t)$ are provided by considering the function $C(t)$ defined in Eq.\ \eqref{Cta}. Clearly, the sample paths are more variable around the sample mean when $\sigma$ is larger.
				\begin{figure}[h]
					\centering
					\subfigure[]{\includegraphics[scale=0.36]{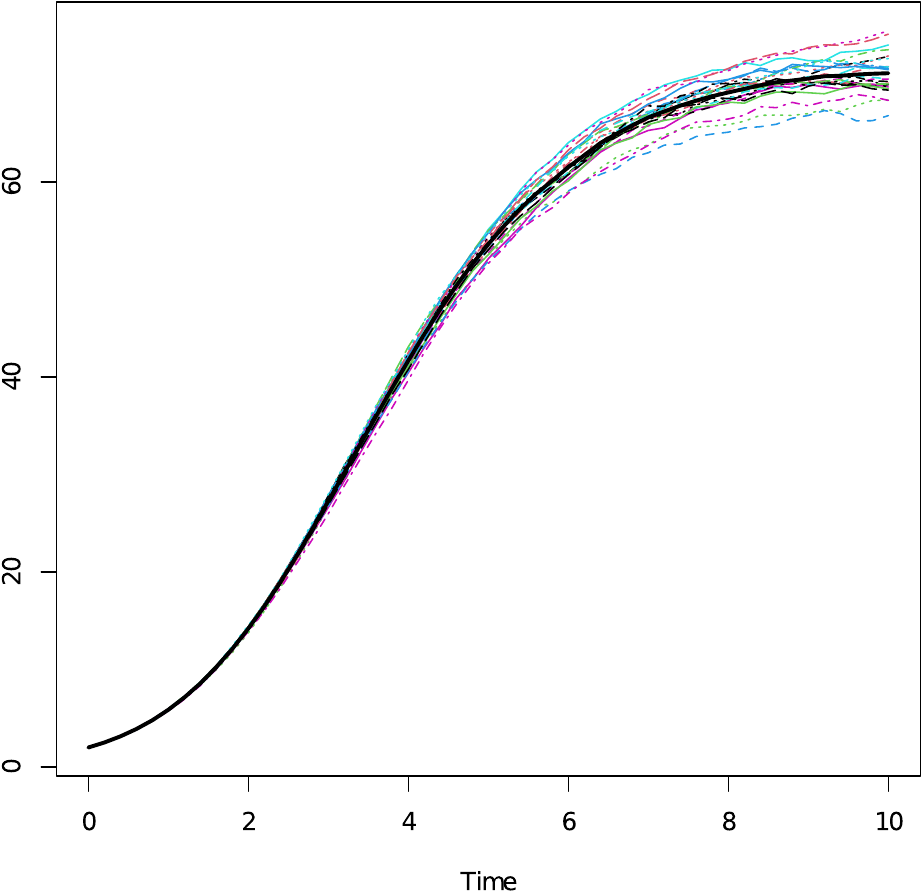}}
					\subfigure[]{\includegraphics[scale=0.36]{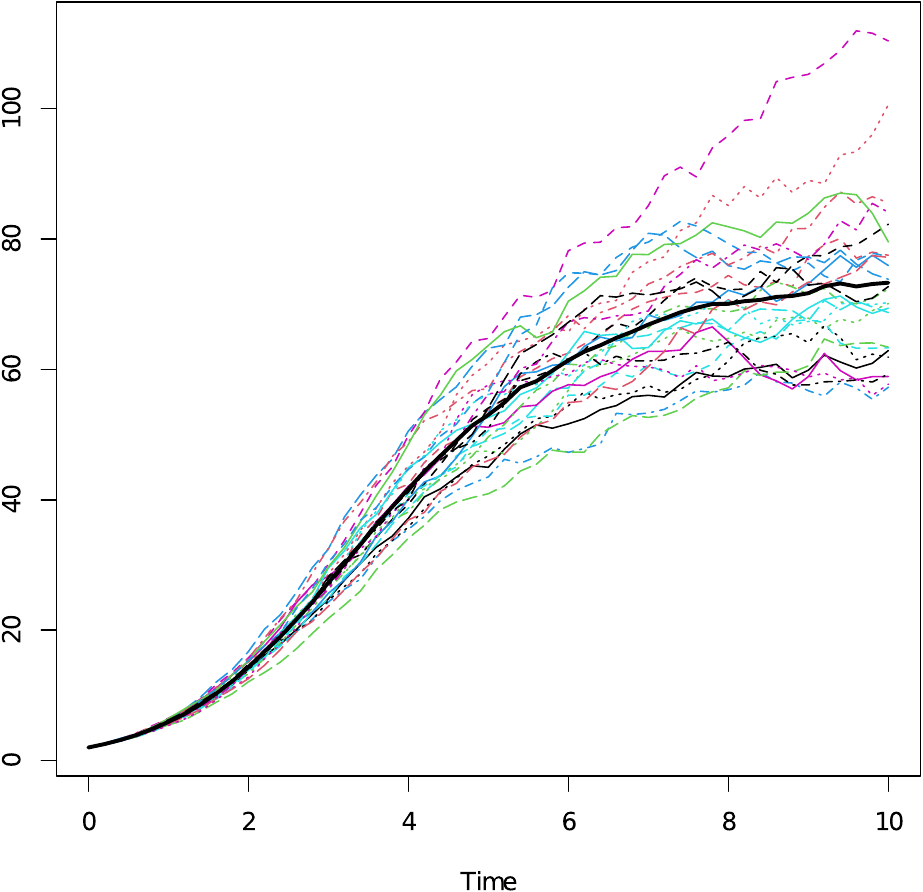}}\\
					\subfigure[]{\includegraphics[scale=0.36]{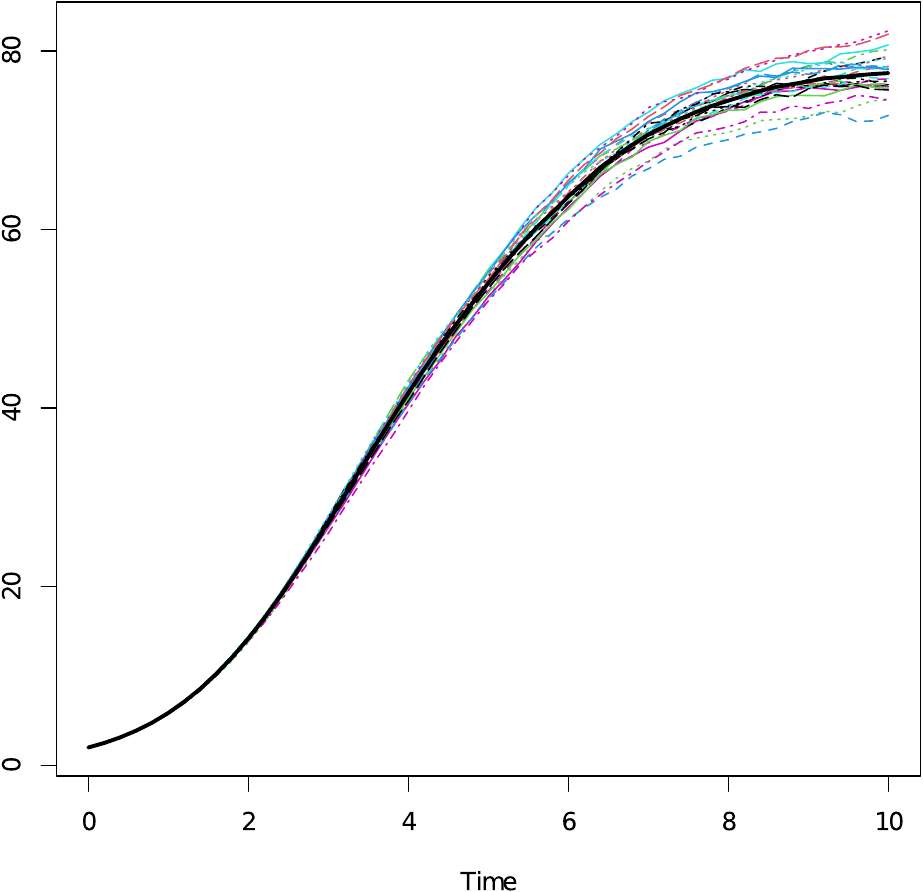}}
					\subfigure[]{\includegraphics[scale=0.36]{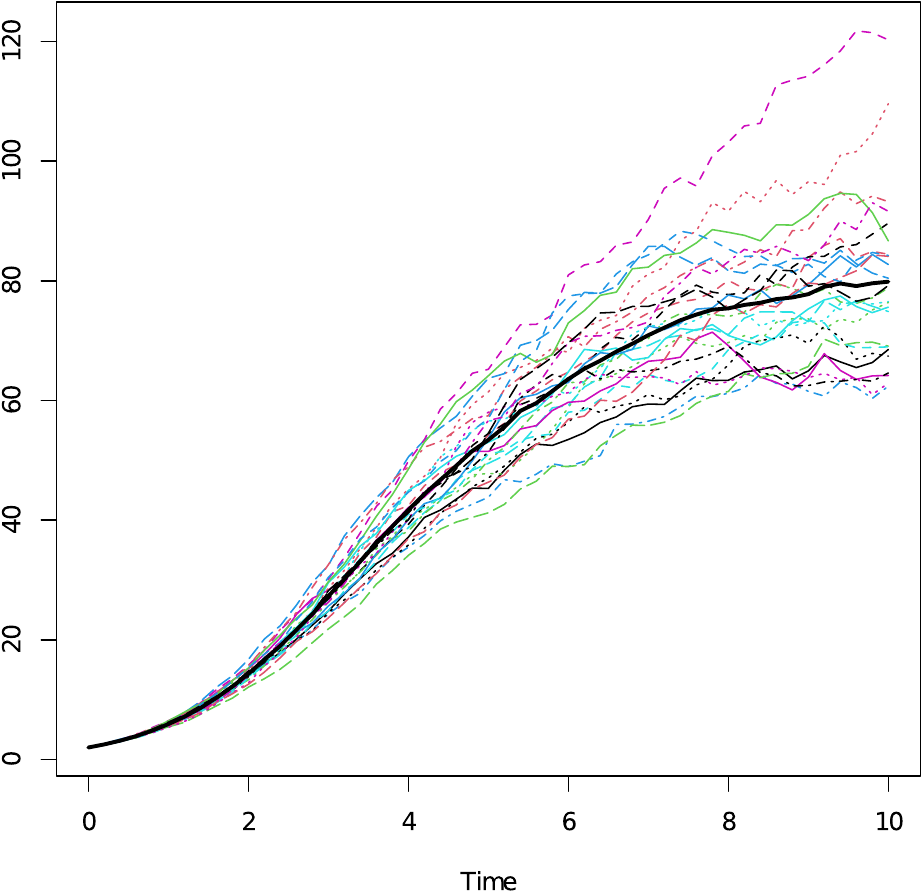}}\\
					\caption{25 simulated sample paths of $X(t)$ with (a) $\sigma=0.01$ and (b) $\sigma=0.05$ and of $\widetilde X(t)$ with $C(t)$ defined in Eq.\ \eqref{Cta}, $p=0.5$, $m=1$, (c) $\sigma=0.01$ and (d) $\sigma=0.05$. In all cases $t_0=0$, $x_0=2$, $\eta=0.2$, $k=0.5$, $q=2$. The black line is the sample mean.}
					\label{fig:Figure6}
				\end{figure}
				
				By developing a well-known strategy (cf.\ Rom\'an-Rom\'an  \textit{et al.}\ (2018) \cite{RomanTorres2018}), it is possible to get the probability distribution of $\widetilde X(t)$. More in detail, if $\widetilde X_0$ follows a lognormal distribution $\Lambda_1(\mu_0,\sigma_0^2)$ or if $\widetilde X_0$ is a degenerate random variable (i.e.\ $\mathsf{P}[\widetilde X_0=x_0]=1$, with $x_0>0$), then the finite-dimensional distributions of the process are lognormal. Indeed, fixing $n$ time instants $t_1<\dots<t_n$, the vector $\left(X(t_1),\dots, X(t_n)\right)^T$ is distributed according to a $n$-dimensional lognormal distribution $\Lambda_n(\epsilon, \Sigma)$, where $\epsilon=(\epsilon_1,\dots,\epsilon_n)^T$ with
				\begin{equation*}
					\epsilon_i=\mu_0+\widetilde H_\xi(t_0,t_i)=\mu_0+q\log \frac{k^{t_0}+\eta}{k^{t_i}+\eta}+\int_{t^*}^{\max\{t_i,t^*\}}C(u)\frac{k^u\mid \log{k}\mid}{\eta+k^u}\mathrm d u-\frac{\sigma^2}{2}(t_i-t_0),\quad i=1,\dots,n
				\end{equation*}
				and $\Sigma=(\sigma_{ij})$ with
				\begin{equation*}
					\sigma_{ij}=\sigma_0^2+\sigma^2\left(\min(t_i,t_j)-t_0\right),\qquad i,j=1,\dots,n.
				\end{equation*}
				From the case $n=2$, it is possible to obtain the conditional probability \textcolor{blue}{distribution} of the process $\widetilde X(t)$, i.e.\
				\begin{equation}\label{transproc}
					\begin{aligned}
						&\left[\widetilde X(t)\mid \widetilde X(s)=x\right]\\
						&\sim\Lambda_1\left(\log x+q\log \frac{k^s+\eta}{k^t+\eta}+\int_{\max\{s,t^*\}}^{\max\{t,t^*\}} C(u)\frac{k^u \mid \log{k}\mid}{\eta+k^u}\textrm{d}u-\frac{\sigma^2}{2}(t-s),\sigma^2(t-s) \right),
					\end{aligned}
				\end{equation}
				for $t>s\ge t_0$.
				By employing Eq.\ \eqref{transproc}, in Table \ref{tab:Table1} we provide  some of the most relevant characteristics of the process making use of the following auxiliary function
				\begin{equation}\label{Gtilde}
					\widetilde G^{\lambda}(t\mid y, \tau):= \exp\left((y+\widetilde H_\xi(\tau,t))\lambda_1+\lambda_2(\lambda_3\sigma_0^2+\sigma^2(t-\tau))^{\lambda_4}\right),
				\end{equation}
				with $\lambda:=(\lambda_1,\lambda_2,\lambda_3,\lambda_4)^T$ and $\widetilde H_\xi(t)$ defined in Eq.\ \eqref{Htilde}.
				\begin{table}[h]
					\caption{The conditional and unconditional mean, mode and $\alpha$-percentile of $\widetilde X(t)$ where $\widetilde G^{\lambda}$ is defined in Eq.\ \eqref{Gtilde} and $z_{\alpha}$ is the $\alpha$-percentile of the standard normal random variable.}
					\label{tab:Table1}
					\centering
					\small
					\begin{tabular}{lll}
						\textbf{Characteristic}  &\textbf{Expression}  & $\lambda$  \\ \hline
						$\mathsf{E}\left[ \widetilde X(t)^n \mid  \widetilde X(s)=y\right]$  &$\widetilde G^{\lambda} (t\mid \log y,s)$  &$(n,n^2/2,0,1)^T$  \medskip \\
						$\mathrm{Mode} \left[ \widetilde X(t) \mid  \widetilde X(s)=y\right]$ &$\widetilde G^{\lambda} (t\mid \log y,s)$  & $(1,-1,0,1)^T$    \medskip \\
						$C_{\alpha}\left[\widetilde X(t)^n \mid  \widetilde X(s)=y\right]$ &$\widetilde G^{\lambda} (t\mid \log y,s)$  &$(1,z_{\alpha},0,1/2)^T$    \medskip \\
						$\mathsf{E} \left[ \widetilde X(t)^n\right]$  &$\widetilde G^{\lambda} (t\mid \mu_0,t_0)$  &$(n,n^2/2,1,1)^T$  \medskip \\
						$\mathrm{Mode} \left[\widetilde X(t) \right]$ &$\widetilde G^{\lambda} (t\mid \mu_0,t_0)$  & $(1,-1,1,1)^T$    \medskip \\
						$C_{\alpha}\left[ \widetilde X(t)^n \right]$ &$\widetilde G^{\lambda} (t\mid \mu_0,t_0)$  &$(1,z_{\alpha},1,1/2)^T$    \medskip \\
					\end{tabular}
				\end{table}
				
				Taking into account the expression of the quantities given in Table \ref{tab:Table1}, we now get some relations between the diffusion processes $\widetilde X(t)$ and $X(t)$. Let
				\begin{equation*}
					G^{\lambda}(t\mid y, \tau):=\exp\left((y+H_\xi(\tau,t))\lambda_1+\lambda_2(\lambda_3\sigma_0^2+\sigma^2(t-\tau))^{\lambda_4}\right),
				\end{equation*}
				with $H_\xi(t)$ defined in Eq.\ \eqref{Htilde1}, be the corresponding auxiliary function for the process $X(t)$. After some algebra, we have that
				\begin{equation}\label{ratio}
					\frac{\mathsf{E}[\widetilde X(t)]}{\mathsf{E}[X(t)]}=\exp\left(\int_{t^*}^{\max\{t,t^*\}}C(s) \frac{k^s\mid \log{k}\mid}{\eta+k^s}\mathrm d s\right), \qquad t\ge t_0,
				\end{equation}
				and
				\begin{equation*}
					\frac{\mathsf{Var}[\widetilde X(t)]}{\mathsf{Var}[X(t)]}=\exp\left(2\int_{t^*}^{\max\{t,t^*\}}C(s) \frac{k^s\mid \log{k}\mid}{\eta+k^s}\mathrm d s\right), \qquad t\ge t_0.
				\end{equation*}
				The above mentioned relations are useful also for estimation purposes. Indeed, in Section \ref{parest} we employ Eq.\ \eqref{ratio} for estimating  the unknown function $C(t)$.
				\section{Parameters estimation}\label{parest}
				The model proposed in Section \ref{DiffProc} is useful in real applications, in particular to model a Bertalanffy-Richards type growth with external modifications. To this aim, an estimation of the unknown parameters of the model is necessary. Since the distribution of $\widetilde X(t)$ is available in explicit form, we propose to obtain the maximum likelihood estimates (MLEs) of the parameters, following the same strategy of Di Crescenzo \textit{et al.}\ (2022) \cite{DiCrescenzoetal2022}.
				In the following reasoning, the time instant $t^*$ is supposed to be known, otherwise $t^*$ needs to be estimated following the strategies described in Section \ref{step2}.
				Let us consider a discrete sampling of the diffusion process $\widetilde X(t)$ consisting of $d$ independent sample paths, with $n_i$ observations for the $i$-th sample path. Hence, the observation times are denoted by  $t_{ij}$, for $i=1,\dots, d$ and $j=1,\dots, n_i$. We also assume, for simplicity, that the first observation time is the same for any trajectory, i.e. $t_{i1}=t_0$, $i=1,\dots,d$. Moreover, let $\widetilde{\mathbb X}=( \widetilde{\bold{X}}^T_1 \mid \dots \mid \widetilde{\bold{X} }^T_d)^T$ be the \textcolor{blue}{vector} which contains all the observed states, with $ \widetilde {\bold{X}}_i=(\widetilde X_{i1},\dots, \widetilde X_{in_i})^T=(\widetilde X(t_{i1}),\dots, \widetilde X(t_{in_i}))^T$. In agreement with the assumptions about $\widetilde X_0$ given in Section \ref{DiffProc}, we suppose that $\widetilde X(t_0)$ follows a lognormal distribution $\Lambda_1(\mu_1,\sigma_1^2)$ with $\mu_1\in\mathbb R$ and $\sigma_1^2\in\mathbb R_+$, so that the density of $\widetilde{\mathbb X}$ is given by
				\begin{equation}\label{dens}
					f_{\widetilde{\mathbb X}}(x)=\prod_{i=1}^{d} \frac{\exp\left(-\frac{(\log x_{i1}-\mu_1)^2}{2\sigma_1^2}\right)}{x_{i1}\sigma_1\sqrt{2\pi}}\prod_{j=1}^{n_i-1}\frac{\exp\left(-\frac{(\log(x_{i,j+1}/x_{ij})-\widetilde m_{\xi}^{i,j+1,j})^2}{2\sigma^2\Delta_i^{j+1,j}}\right)}{x_{ij}\sigma\sqrt{2\pi \Delta_i^{j+1,j}}},
				\end{equation}
				where $x=(x_{1,1},\dots, x_{1,n_1}\mid \dots \mid x_{d1},\dots, x_{d,n_d})^T\in \mathbb R_+^{n+d}$, with  $n=\displaystyle \sum_{i=1}^{d}(n_i-1)$, for (cf.\ \eqref{Htilde})
				\begin{equation*}
					\widetilde m_{\xi}^{i,j+1,j}:=\widetilde H_\xi(t_{ij},t_{i,j+1})=q \log \frac{k^{t_{ij}}+\eta}{k^{t_{i,j+1}}+\eta}+\int_{\max\{t_{ij},t^*\}}^{\max\{t_{i,j+1},t^*\}}C(u)\frac{k^u\mid\log{k}\mid}{\eta+k^u}\mathrm d u-\frac{\sigma^2}{2}(t_{i,j+1}-t_{ij}),
				\end{equation*}	
				and $\Delta_{i}^{j+1,j}:=t_{i,j+1}-t_{ij}$ for $i=1, \dots, d$ and $j=1,\dots, n_i$.
				In order to simplify the cumbersome expression of the density given in Eq.\ \eqref{dens}, we perform the following change of variables:
				\begin{equation*}
					\begin{aligned}
						&V_{0i}:=\widetilde X_{i1},\qquad\qquad i=1,\dots,d \\
						&V_{ij}:=\left(\Delta_i^{j+1,j}\right)^{-1/2}\log\frac{\widetilde X_{i,j+1}}{\widetilde X_{ij}}, \qquad \qquad j=1,\dots,n_i-1,\quad i=1,\dots,d.
					\end{aligned}
				\end{equation*}	
				Therefore, by setting $\mathbb V=(\bold V_0^T\mid \bold V_1^T\mid \dots\mid \bold V_d^T)^T$ with $\bold V_i^T=(V_{i1},\dots, V_{in_i})$, the corresponding density is
				\begin{equation*}
					f_{\mathbb V}(v)=\frac{\exp\left(-\frac{1}{2\sigma_1^2}(lv_0-\mu_1\mathbb I_d)^T(lv_0-\mu_1\mathbb I_d)\right)}{\prod_{i=1}^d v_{0i}(2\pi\sigma_1^2)^{d/2}}\cdot \frac{\exp\left(-\frac{1}{2\sigma^2}(v_{(1)}-\gamma^{\xi})^T(v_{(1)}-\gamma^{\xi})\right)}{(2\pi\sigma^2)^{n/2}},
				\end{equation*}
			\textcolor{blue}{
				where
				$$
				\begin{aligned}
				&v=(v_0^T|v_{(1)}^T)\in\mathbb R^{n+d},\quad v_0=(v_{01},\dots,v_{0d})^T\in\mathbb R^{d},\quad lv_0=(\log v_{01},\dots,\log v_{0d})^T,\\
				&v_{(1)}=(v_{11},\dots,v_{1,n_1-1}|\dots|v_{d1},\dots,v_{d,n_d-1})^T\in\mathbb R^{n},\quad \mathbb I_d=(1,\dots,1)^T\in \mathbb R^d,
				\end{aligned}
				$$
				 with $\gamma^{\xi}=(\gamma_{11}^\xi,\dots, \gamma_{1,n_1-1}^\xi, \dots, \gamma_{d1}^\xi, \dots,\gamma_{d,n_d-1}^\xi)^T\in\mathbb R^n$ and $\gamma_{ij}^\xi=\left(\Delta_i^{j+1,j}\right)^{-1/2}\widetilde m_{\xi}^{i,j+1,j}$, for $j=1,\dots, n_i-1$ and $i=1,\dots,d$. }\\
				By supposing that the parameters of the process, $\xi$, and those of the initial distribution, $(\mu_1,\sigma_1^2)$, are functionally independent, the log-likelihood function can be expressed as follows
				\begin{equation*}
					L_{\mathbb V}(\xi,\mu_1,\sigma_1^2)
					=\widetilde{L}_{\mathbb V}(\xi)-\frac{n+d}{2}\log{2\pi}-\frac{d}{2}\log{\sigma_1^2}
					-\sum_{i=1}^{d}\log{v_{0i}}-\frac{\sum_{i=1}^d (\log v_{0i}-\mu_1)^2}{2\sigma_1^2},
				\end{equation*}
				where
				\begin{equation}\label{objfun}
					\widetilde{L}_{\mathbb V}(\xi):=-\frac{n}{2}\log\sigma^2-\frac{Z_1+\Phi_\xi-2\Gamma_\xi}{2\sigma^2}
				\end{equation}
				and
				\begin{equation*}
					Z_1:= \sum_{i=1}^{d} \sum_{j=1}^{n_i-1}v_{ij}^2,\quad \Phi_\xi:= \sum_{i=1}^{d} \sum_{j=1}^{n_i-1} \frac{(\widetilde m_\xi^{i,j+1,j})^2}{\Delta_i^{j+1,j}},\quad \Gamma_\xi:= \sum_{i=1}^{d} \sum_{j=1}^{n_i-1}\frac{v_{ij}\widetilde m_\xi^{i,j+1,j}}{(\Delta_i^{j+1,j})^{1/2}}.
				\end{equation*}
				By performing the partial derivatives of the log-likelihood $L_{\mathbb V}$ with respect to $\mu_1$ and $\sigma_1^2$, we can easily get the following MLEs:
				\begin{equation}\label{MLEini}
					\widehat\mu_1=\frac{1}{d}\sum_{i=1}^d\log v_{0i},\qquad \widehat\sigma_1^2=\frac{1}{d}\sum_{i=1}^d(\log v_{0i}-\widehat \mu_1)^2.
				\end{equation}
				In order to obtain the MLEs of $\xi$, we decide to adopt suitable metaheuristic optimization methods to avoid numerical problems in the resolution of the system of maximum likelihood.
				In addition, we need to estimate also the unknown function $C(t)$ occurring in the definition of the modified model $\widetilde X(t)$. To this aim, we assume to dispose of a data set concerning the observations of a modified Bertalanffy-Richards model $\widetilde X(t)$ in the time interval $[t_0,T]$. Consider first Eq.\ \eqref{ratio}, which allows us to implement the following procedure.
				\begin{framed}\label{procedure}
					\textbf{Procedure 1 - Estimation of the parameters}
					\begin{itemize}
						\item[\small\textbf{Step 1}] Compute the MLEs $\hat \xi=(\hat q, \hat k, \hat \eta, \hat \sigma)^T$.
						\item[\small\textbf{Step 2}] Determine an estimation of $t^*$.
						\item[\small\textbf{Step 3}]  Obtain an estimation of the function $C(t)$.
					\end{itemize}
				\end{framed}
				A description of the steps of Procedure 1 is provided in the hereafter.\\
				
				\textbf{Step 1} In order to determine the MLEs $\hat \xi$, we use the data over an interval $I_t=[t_0, \bar t]\supset[t_0,t_I]$, being $\bar t$ a time instant quite close to $t^*$. For example, we can take $\bar t$ as the time instant in which the spline $S(t)$ interpolating the mean of the paths intercepts the threshold $S$, defined in Eq.\ \eqref{Sdef}. \\
				
				\textbf{Step 2} To determine an estimation of $t^*$ two different strategies are available:
				(i) we can use Eq.\ \eqref{tstar} and consider $t^*$ as a parametric function of the MLEs determined at Step 1 or (ii) we can consider the mean of the first-passage-time of the estimated process (i.e.\ the process given in Eq.\ \eqref{difproc1} obtained by considering the MLEs of the parameters) through the fixed boundary $S$. \\

				\textbf{Step 3} Considering the sample mean $\mathsf{E}[\widetilde X(t)]$ of the modified process \eqref{difproc} and the estimated mean $\widehat{\mathsf{E}} \left[X(t)\right]$ of the classical process over the full time interval $[t_0, T]$, we then obtain the estimation of the function $C(t)$ from Eq.\ \eqref{ratio} as follows
					\begin{equation}\label{Cstim}
						\widehat C(t)=\frac{\eta+k^t}{k^t\mid \log{k}\mid}\; \frac{\textrm{d}}{\textrm{d}t}m(t),\qquad t>t^*,
					\end{equation}
					with $m(t):=\log\left(\frac{{\mathsf{E}}[\widetilde X(t)]}{\widehat{\mathsf{E}}[X(t)]}\right)$, for $t>t^*$.
				Note that since the mean of the process $X(t)$ (cf.\ Eq.\  \eqref{difproc1}) corresponds to the classical deterministic curve $x_\theta(t)$ (cf.\ Eq.\ \eqref{BR}),
				in Eq.\ \eqref{Cstim} we consider
				\begin{equation*}
					\widehat{\mathsf{E}} \left[X(t)\right]=\mathsf E [X_0]\left(\frac{\widehat \eta + \widehat k^{t_0}}{\widehat \eta + \widehat k^{t}}\right)^{\widehat q}.
				\end{equation*}
		\section{Determination of MLEs through heuristic optimization methods}\label{SA}
				To determine the MLEs mentioned in the \textbf{Step 1} of \textbf{Procedure 1}, we develop the following reasoning.
				\par
				Since the numerical methods employed to solve system  of maximum likelihood may fail to converge even in the case of a quite accurate initial solution, we now propose to employ two different meta-heuristic optimization methods, namely Simulated Annealing (SA) and Ant Lion Optimizer (ALO). These two methods belong to the family of \textcolor{blue}{gradient-free} algorithms and they are suggested when the use of the derivative of the function to be maximized is intricate. The estimation of the parameters of the initial distribution can be obtained through the explicit expressions given in Eq.\ \eqref{MLEini}. Hence, we need to maximize Eq.\ \eqref{objfun} on the parametric space $\Theta=\{(q,k,\eta,\sigma): q>0,\, 0<k<1,\eta>0,\, \sigma>0\}$. Since this space is continuous and unbounded, a restriction of $\Theta$ is needed. A similar problem has been considered in Section 3.1.1 of Rom\'an-Rom\'an and Torres-Ruiz (2015) \cite{Romanetal2015}. Their main idea is to find, first of all, a bounded interval $I_q:=[q_1,q_2]$ for $q$ and from that determine sequentially two bounded intervals, the former for $k$ and the latter for $\eta$. Further on, as usual, we use the interval  $I_\sigma:=(0,0.1)$ for the parameter $\sigma$ so that the paths are compatible with a Bertalanffy-Richards type growth.
				\begin{itemize}
					\item[\textbf{1.1}]
					Let us start from the determination of $I_q$. By interpolating the mean of the sample paths with a natural cubic spline $S(t)$, we find the time instant $t_I^*$ which is the  ``observed" inflection instant. Then, we select $t_1$ as the first instant in which $S(t)/\mathcal{K}^*>e^{-1}$, where $\mathcal{K}^*$ is an approximation of the boundary $S=(1+p)x_i^*$ being $x_i^*$ the value of the curve at the observed inflection time. In this way, $\mathcal{K}^*$ may represent an approximation of the carrying capacity of the classical Bertalanffy-Richards model. Thus, we consider the interval $[t_1,t_2]$ where $t_2$ is the observation instant immediately after $t_I^*$. Consequently, we determine the \textcolor{blue}{endpoints} of the interval $I_q=[q_1,q_2]$ being $q_j$ the solutions of the following equations
					\begin{equation}\label{eq-q}
						S(t_j)=\mathcal{K}^* \left(\frac{q_j}{1+q_j}\right)^{q_j},\qquad \qquad j=1,2.	
					\end{equation}
					\item[\textbf{1.2}] We now proceed on the determination of $I_k$ and $I_\eta$.
					From the characteristics of the classic Bertalanffy-Richards curve, we have
					\begin{equation*}
						\eta=qk^{t_I},\qquad \qquad k=\left(q\left(\left(\frac{\mathcal{K^*}}{x_0}\right)^{1/q}-1\right)\right)^{1/(t_0-t_I)}.
					\end{equation*}
					This allows us to consider the following functions
					\begin{equation}\label{eq-k}
						g(t,q)=\left(q\left(\left(\frac{\mathcal{K^*}}{x_0}\right)^{1/q}-1\right)\right)^{1/(t_0-t)},\quad h(q,k,t)=qk^t,
					\end{equation}
					for $q\in I_q,\;t\in[t_1,t_2],\; k\in I_k\subset(0,1)$.
					From Eq.\ \eqref{eq-k}, it follows that the minimun value of $g(t,q)$ is reached at $(t_1,q_1)$ and its maximum at $(t_2,q_2)$. Hence, we consider for the parameter $k$ the interval  $I_k:=[k_1,k_2]=[g(t_1,q_1),g(t_2,q_2)]$. Similarly, the function $h(q,k,t)$ assumes its minimum value at $(q_1,k_1,t_2)$ and its maximum at $(q_2,k_2,t_1)$, so that for the parameter $\eta$ we can use the interval $I_{\eta}:=[\eta_1,\eta_2]=[h(q_1,k_1,t_2),h(q_2,k_2,t_1)]$, since the function $h(q,k,t)$ is decreasing with respect to $t$.
				\end{itemize}
				It is worth to notice that the width of the interval $I_q$ may be too large, since small variations of the ratio $\frac{S(t_j)}{\mathcal K^*}$ correspond to large variations of the parameter $q$ (cf.\ Eq.\ \eqref{eq-q}), as can be deducted from Figure \ref{fig:Figure10}.
				\begin{figure}[h]
					\centering
					\includegraphics[scale=0.4]{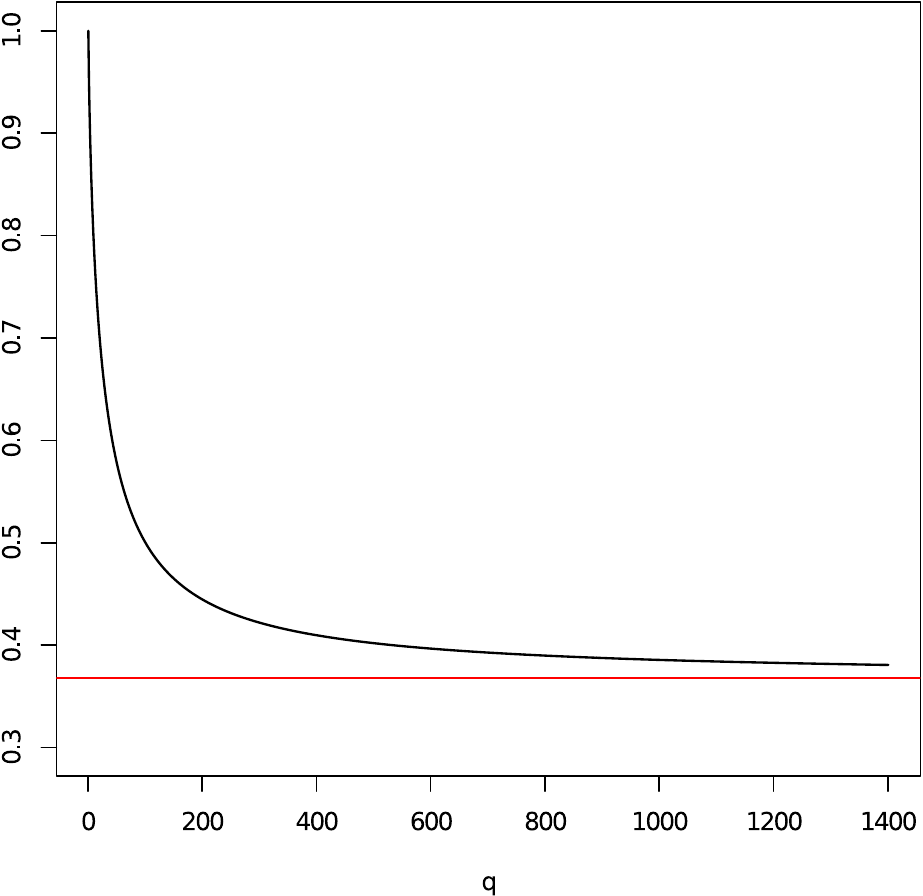}
					\caption{The ratio $\frac{S(t_j)}{\mathcal K^*}$ as a function of $q$ (black) and the line $y=e^{-1}$ (red).}
					\label{fig:Figure10}
				\end{figure}
					Nevertheless, the resulting MLEs obtained via SA and ALO are quite accurate as can be noticed from the results shown in Section \ref{Simulation}.
				\section{Simulation}\label{Simulation}
				In this section, we devote our analysis to a simulation study in order to validate the procedures described in Sections \ref{parest} and \ref{SA}. The pattern of the simulations is based on $25$ sample paths of $\widetilde X(t)$, defined in Eq.\ \eqref{difproc}, the time interval $[0,10]$ with $q=2$, $k=0.5$, $\eta=0.2$, $\sigma\in\{0.01,0.02\}$.
 Moreover, $C(t)$ is taken as in Eq.\ \eqref{Cta} for $m=1$ and where the value of $t^*$ is obtained by setting
				 $p=0.5$ in Eq.\ (\ref{tstar}).
				The sample paths share the same length since the common observation time instants are taken as $t_j=j\cdot 0.1$ with $j=0,\dots,100$. For the initial condition we consider a degenerate distribution centered in $x_0=2$. In the study, we suppose that the threshold $S=(1+p)x(t_I)$ is known, in the sense that the value of $p$ is given. Otherwise, a set of possible values for $p$ can be explored to select the best estimation (in terms of minimization of the absolute relative error between the sample mean and the estimated mean).
Once the sample paths are simulated, we consider a sample size $n=51$, being the data equally spaced in the interval under consideration.
				With the aim of obtaining better estimates, the steps of \textbf{Procedure 1} are replicated  $300$ times. For any parameter, the mean of the resulting estimates is taken as the final value. As can be seen in Figure \ref{fig:Figure17marzo}, the estimates stabilize around a specific value as the number of replications increases.
				
				\textcolor{blue}{
				For the realization of the simulation study and the application to real data, the \texttt{R} software has been used. Specifically, the \texttt{metaheuristicOpt} package, which implements the metaheuristic Ant Lion Optimizer algorithm, and the \texttt{GenSA} package, which contains the Simulated-Annealing algorithm. For both methods, the intervals, determined in Section 7, were given in input. Then, the initial solution of each parameter was  determined by choosing a random value with a uniform distribution in any interval. The machine used for the computations was an Apple M1, with a total of 8 cores. The time taken to obtain the results reported below with 300 replications is 1800.325s, and with 100 replications is 582.049s.
			}
				
				\begin{figure}[h]
					\centering
					\subfigure[]{\includegraphics[scale=0.35]{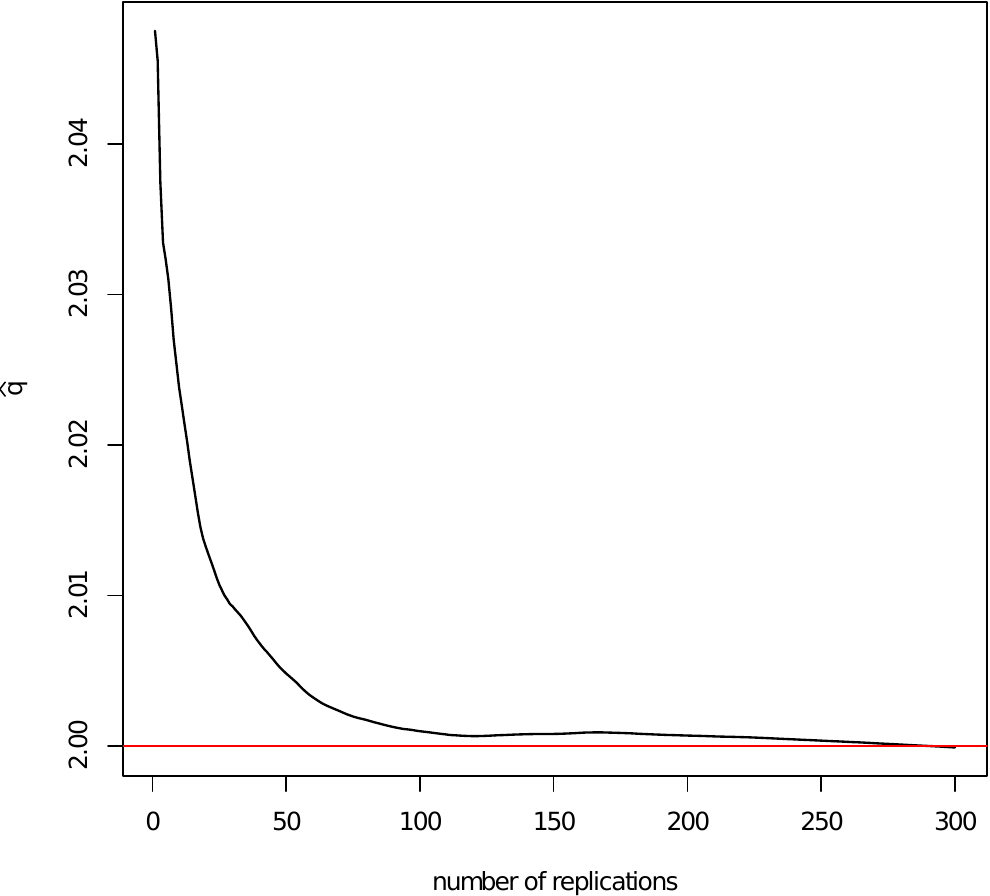}}\qquad
						\subfigure[]{\includegraphics[scale=0.35]{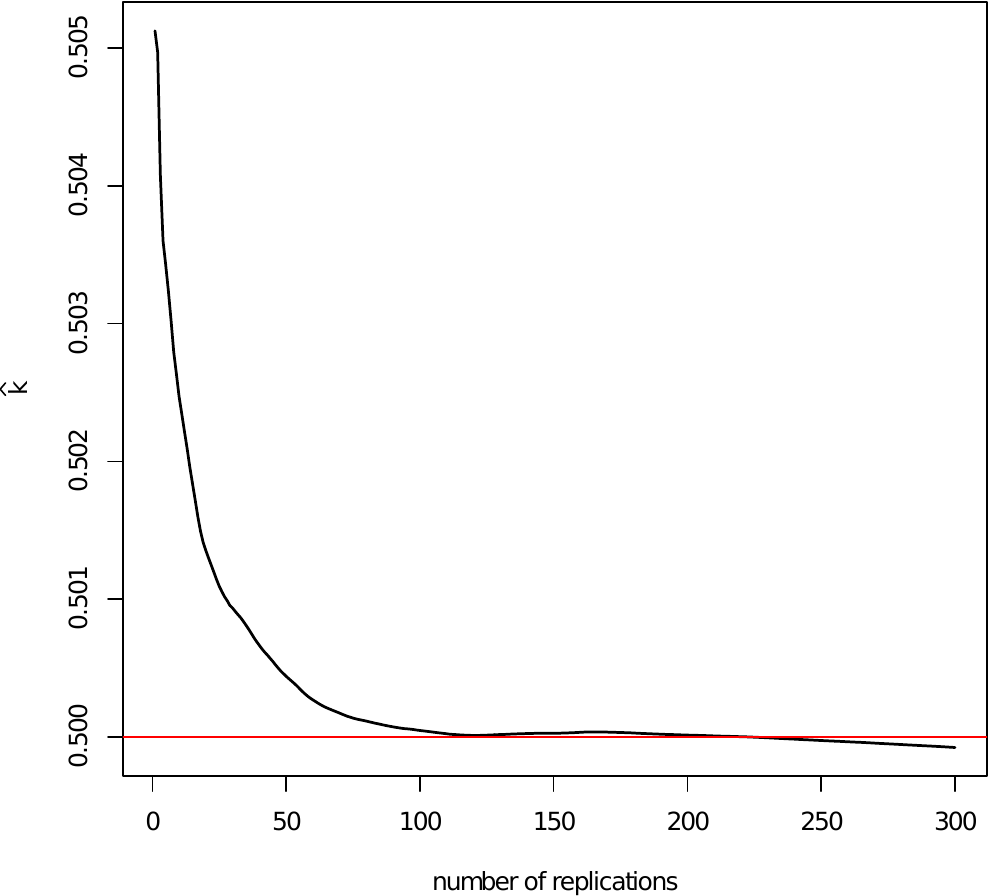}}\qquad
							\subfigure[]{\includegraphics[scale=0.35]{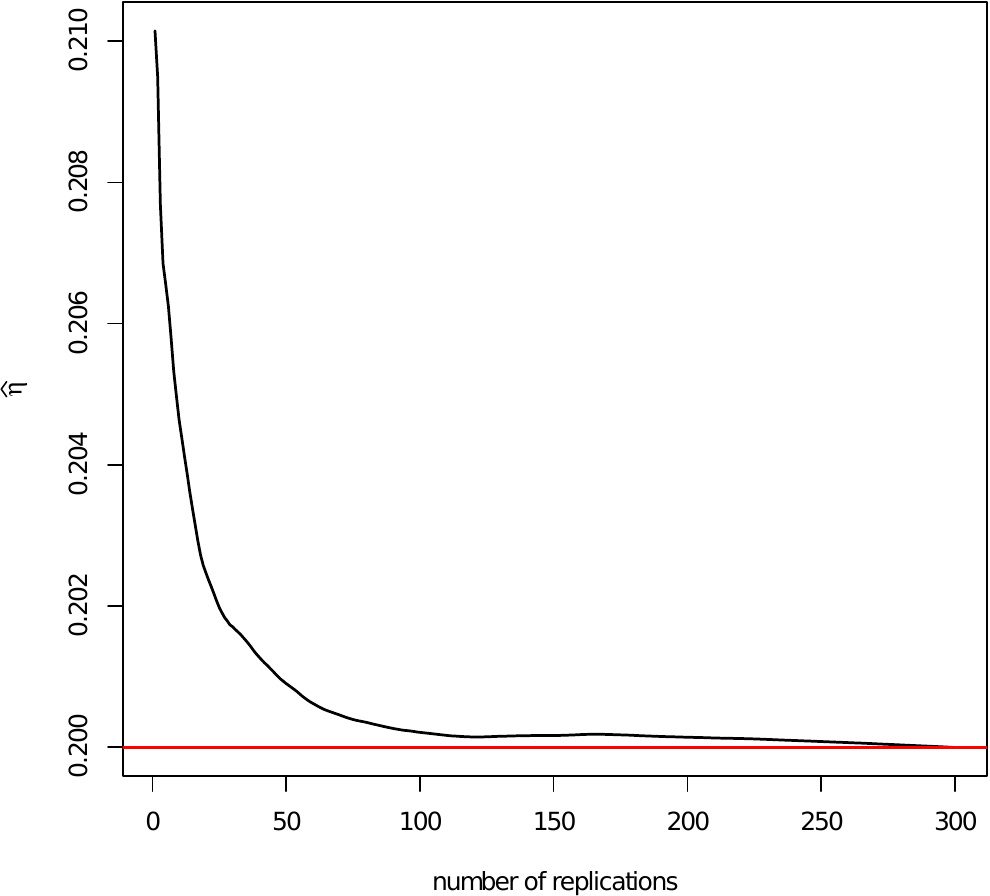}}\qquad
								\subfigure[]{\includegraphics[scale=0.35]{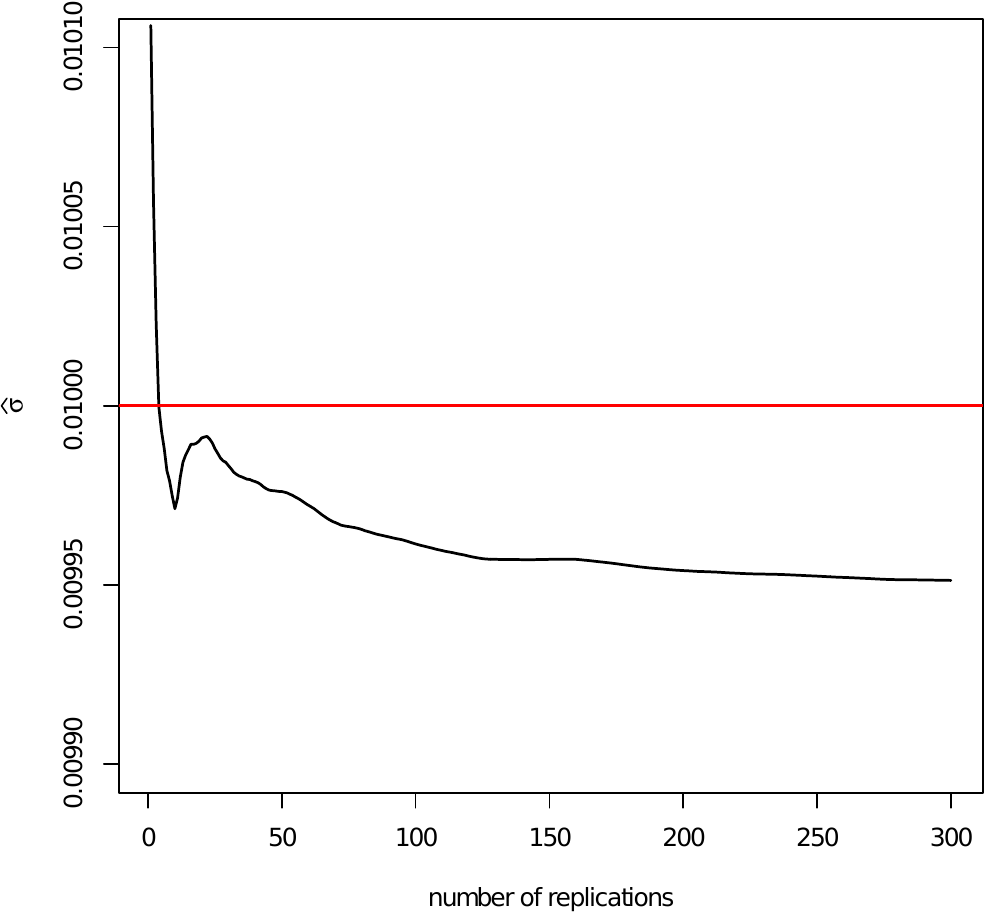}}
					\caption{Convergence of the estimates of the parameter (a) $\widehat q$, (b) $\widehat k$, (c) $\widehat \eta$, (d) $\widehat \sigma$ in the case $\sigma=0.01$. The red lines represent the real parameters.}
					\label{fig:Figure17marzo}
				\end{figure}

				Let us now focus on the estimation of the  parameters $q,k,\eta$, $\sigma$ and of the function $C(t)$.
				\subsection{Step 1: Estimate of $\xi$}\label{Step1}
				For any $j=0,\dots,50$ we approximate the sample mean of $\widetilde x_i (t_j)$, $i=1,\dots,25$, with a natural cubic spline $S(t)$. Then, through its derivatives, we also determine an approximation of the observed inflection time instant.  See Figure \ref{fig:Figure7} for the plot of $\frac{\textrm{d}}{\textrm{d}t} S(t)$ and $\frac{\textrm{d}^2}{\textrm{d}t^2} S(t)$, and Table \ref{tab:Tabella1} for a comparison between the theoretical inflection point $t_I$ and observed inflection instant.
				\begin{figure}[h]
					\centering
					\subfigure[]{\includegraphics[scale=0.4]{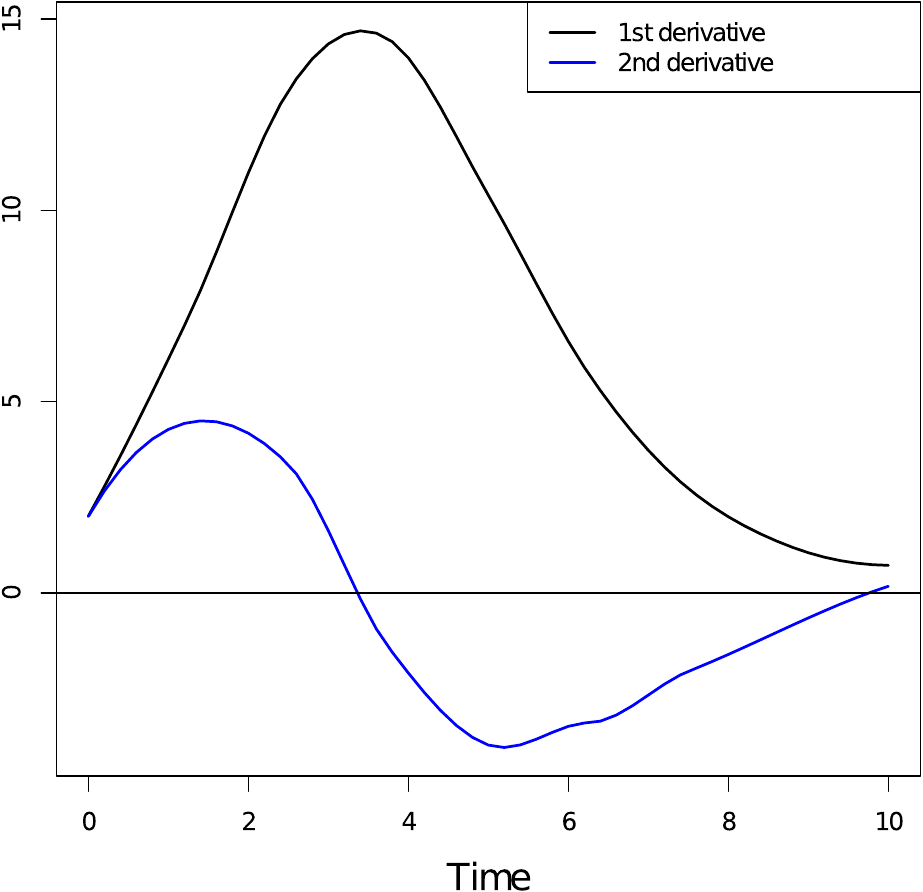}}\qquad
					\subfigure[]{\includegraphics[scale=0.4]{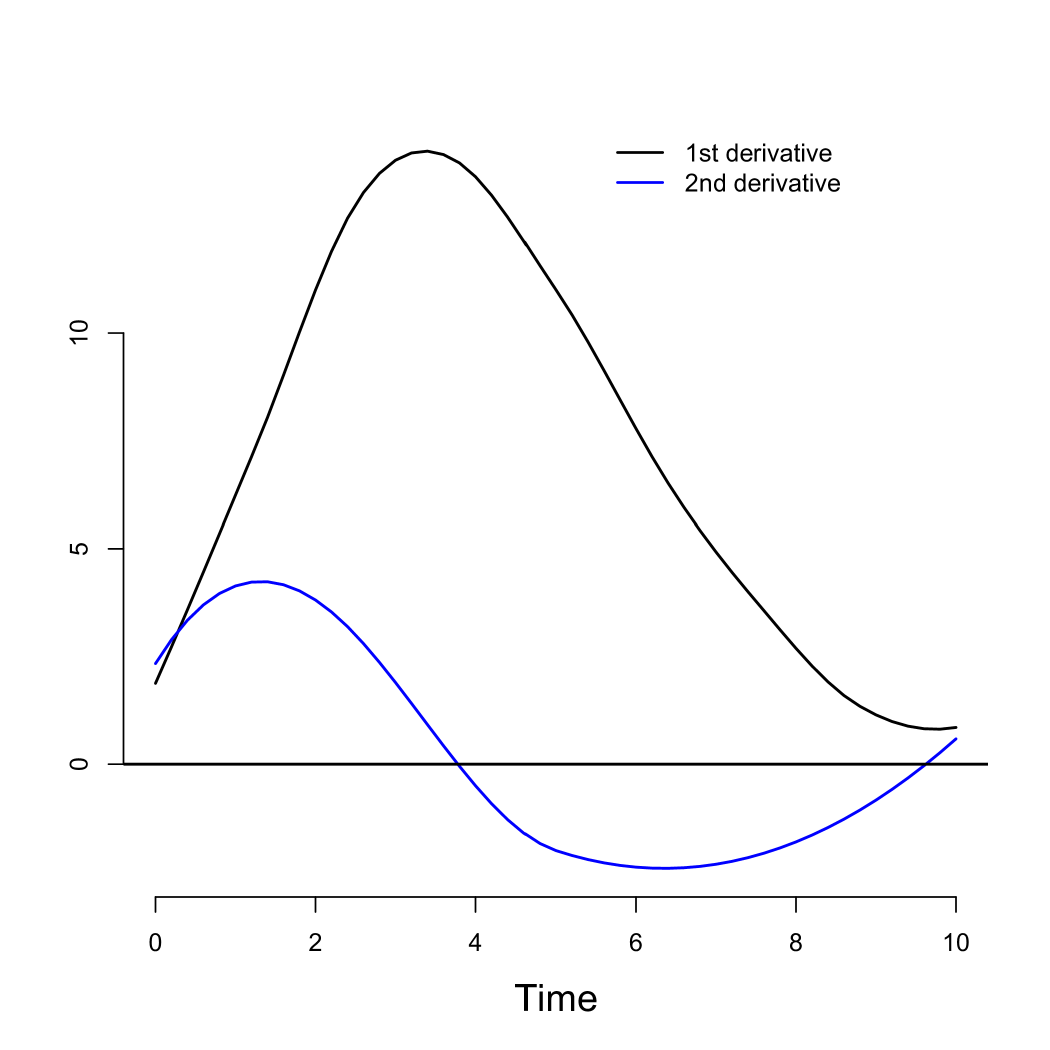}}
					\caption{The $1$st and $2$nd derivative of the spline $S(t)$ for (a) $\sigma=0.01$ and (b) $\sigma=0.02$.}
					\label{fig:Figure7}
				\end{figure}

				\begin{table}[h]
					\caption{The theoretical and the observed inflection instant,  and the corresponding absolute relative error (RAE).}
					\label{tab:Tabella1}
					\centering
					\begin{tabular}{c|c|c|c|c}
						parameter              & theoretical inflection instant & $\sigma$ & observed inflection instant & RAE          \\ \hline
						\multirow{2}{*}{$t_I$} & \multirow{2}{*}{$3.32193$}  & $0.01$   & $3.36223$                & $0.01213$ \\ \cline{3-5}
						&                                & $0.02$   & $3.49254$                & $0.05136$
					\end{tabular}
				\end{table}
				To apply SA or ALO algorithm, we first compute the intervals $I_\nu$ with $\nu\in\{q,k,\eta,\sigma\}$ using the strategy described in Section \ref{SA}. We remark that in this first part of the procedure the knowledge of the exact value of $t^*$ is irrelevant since we use the data over the time interval $[t_0,\bar t]$ where $\bar t$ represents the time instant corresponding to the boundary $S=(1+p)x^*_I$. The obtained intervals are given in Table \ref{tab:Tabella2}. After that, we apply SA  and ALO algorithms. 
				The results, given in Table \ref{tab:Tabella2}, show that even if the width of the intervals $I_\nu$, $\nu\in\{q,k,\eta,\sigma\}$, is large, the RAE of the resulting MLEs is small.
				In Table \ref{tab:Tabella15feb} we provide the MLEs obtained by considering different number of replications. We remark that the estimates are still reasonable even when the number of replications is smaller.
				By comparing the RAE of the MLEs obtained via SA and via ALO, we note that the results obtained via SA are better than those obtained via ALO. For this reason, from now on, we adopt the MLEs obtained via SA.
				%
				%
				%
				\begin{table}[]
					\caption{The real values, the minimum and maximum width of the bounded intervals,  the MLEs (determined with SA and ALO algorithm) and the corresponding absolute relative error (RAE) of the parameters.}
					\label{tab:Tabella2}
					\centering
					\small
					\begin{tabular}{l|l|l|l|l|l|l}
						parameter                   & real value                & min width                                 & max width                      & method & MLE       & RAE       \\ \hline
						\multirow{2}{*}{$q$}        & \multirow{2}{*}{$2$}      & \multirow{2}{*}{$4.15554$}   & \multirow{2}{*}{$5.11776$} & SA   & $1.99990$ & $0.00005$ \\ \cline{5-7}
						&                           &                                             &                            & ALO & $2.03645 $ & $0.01822 $ \\ \hline
						\multirow{2}{*}{$k$}        & \multirow{2}{*}{$0.5$}    & \multirow{2}{*}{$0.63553$}    & \multirow{2}{*}{$0.67700$} & SA   & $0.49992$ & $0.00015$ \\ \cline{5-7}
						&                           &                                             &                            & ALO & $0.50340$ & $0.00680 $ \\ \hline
						\multirow{2}{*}{$\eta$}     & \multirow{2}{*}{$0.2$}    & \multirow{2}{*}{$1.47345$} & \multirow{2}{*}{$1.97406$} & SA   & $0.19999$ & $0.00001$ \\ \cline{5-7}
						&                           &                                             &                            & ALO & $0.20703$ & $0.03514 $ \\ \hline
						\multirow{2}{*}{$\sigma$} & \multirow{2}{*}{$0.01$} & \multirow{2}{*}{$0.1$}    & \multirow{2}{*}{$0.1$}    & SA   & $0.00995$ & $0.00487$ \\ \cline{5-7}
						&                           &                                             &                            & ALO & $0.00988$ & $0.01200 $
					\end{tabular}
					\\ \vspace{0.3cm}
						
					\begin{tabular}{l|l|l|l|l|l|l}
						parameter                   & real value                & min width                                    & max width                      & method & MLE       & RAE       \\ \hline
						\multirow{2}{*}{$q$}        & \multirow{2}{*}{$2$}      & \multirow{2}{*}{$3.78010$}   & \multirow{2}{*}{$5.68256$} & SA   & $1.99645$ & $0.00177$ \\ \cline{5-7}
						&                           &                                             &                            & ALO & $2.08720$ & $0.04360$ \\ \hline
						\multirow{2}{*}{$k$}        & \multirow{2}{*}{$0.5$}    & \multirow{2}{*}{$0.60332$}    & \multirow{2}{*}{$0.69750$} & SA   & $0.49942$ & $0.00115$ \\ \cline{5-7}
						&                           &                                             &                            & ALO & $0.50685$ & $0.01370$ \\ \hline
						\multirow{2}{*}{$\eta$}     & \multirow{2}{*}{$0.2$}    & \multirow{2}{*}{$1.36523$} & \multirow{2}{*}{$2.39929$} & SA   & $0.19941$ & $0.00295$ \\ \cline{5-7}
						&                           &                                             &                            & ALO & $0.21742$ & $0.08710$ \\ \hline
						\multirow{2}{*}{$\sigma$} & \multirow{2}{*}{$0.02$} & \multirow{2}{*}{$0.1$}    & \multirow{2}{*}{$0.1$}    & SA   & $0.01993$ & $0.00337$ \\ \cline{5-7}
						&                           &                                             &                            & ALO & $0.01994$ & $0.00293$
					\end{tabular}
					
				\end{table}

				\begin{table}[]
					\caption{The MLEs (determined with SA and ALO algorithm) and the corresponding absolute relative error of the parameters with different numbers of replications. The real values of $q$, $k$, $\eta$ and $\sigma$ are given in the first row of both the tables.}
					\label{tab:Tabella15feb}
					\centering
					\begin{tabular}{llllll}
						Method & No. of replications & $q=2$           & $k=0.5$           & $\eta=0.2$        & $\sigma=0.01$     \\ \hline
						SA     & 30               & $\hat q=2.00928$    & $\hat k=0.50094$    & $\hat\eta=0.20171$   & $\hat\sigma=0.00998$ \\
						& (relative error:)   & ($0.00464$)    & ($0.00187$)    & ($0.00854$)  & ($0.00167$)    \\ \hline
						ALO    & 30               & $\hat q=2.00880$   & $\hat k=0.50098$   & $\hat\eta=0.20166$   & $\hat\sigma=0.00984$ \\
						& (relative error:)   &$(0.00440)$  & $(0.00195)$   &$(0.00831)$  & $(0.01616)$   \\ \hline
						SA     & 300               & $\hat q=1.99990$    & $\hat k=0.49992$   & $\hat\eta=0.19999$   & $\hat\sigma=0.00995$ \\
						& (relative error:)   & ($0.00005$)  & ($0.00015$)  & ($0.00001$)  & ($0.00487$)    \\ \hline
						ALO    & 300               & $\hat q=2.03645$    & $\hat k=0.50340$  & $\hat\eta=0.20703$ & $\hat\sigma=0.00988$ \\
						& (relative error:)   & $(0.00680)$ & $(0.00681)$ & $(0.03514)$ & $(0.01200)$
					\end{tabular}
					\\ \vspace{0.3cm}
					
					\begin{tabular}{llllll}
						Method & No. of replications & $q=2$           & $k=0.5$           & $\eta=0.2$        & $\sigma=0.02$     \\ \hline
						SA     & 30               & $\hat q=2.00871 $    & $\hat k=0.50074$    & $\hat\eta=0.20187$   & $\hat\sigma=0.02001$ \\
						& (relative error:)   & ($0.00435$)    & ($0.00148$)    & ($0.00934$)  & ($0.00031$)    \\ \hline
						ALO    & 30               & $\hat q=2.02218$    & $\hat k=0.50107$   & $\hat\eta=0.20448$   & $\hat\sigma=0.01999$ \\
						& (relative error:)   & $(0.01109)$  & $(0.00215)$    & $(0.02238)$  & $(0.00060)$   \\ \hline
						SA     & 300               & $\hat q=1.99645 $    & $\hat k=0.49942$    & $\hat\eta=0.19941$   & $\hat\sigma=0.01993$ \\
						& (relative error:)   & ($0.00177$)    & ($0.00115$)    & ($0.00295$)  & ($0.00337$)    \\ \hline
						ALO    & 300               & $\hat q=2.08720$    & $\hat k=0.50685$   & $\hat\eta=0.21742$   & $\hat\sigma=0.01994$ \\
						& (relative error:)   & $(0.04360)$ & $(0.01370)$ & $(0.08710)$ & $(0.00293)$
					\end{tabular}
				\end{table}
					
					
					\subsection{Step 2: Estimate of $t^*$}\label{step2}
					In order to compute the estimated time instant $\hat t^*$, two options are available: (i) we can use the MLEs obtained so far and compute $\hat t^*$ by means of the deterministic formula given in Eq.\ \eqref{tstar}, or (ii) we can compute  $\hat t^*$ as the mean of first-passage-time (FPT) of the estimated process (i.e. the process whose parameters are estimated by the MLEs) through the boundary $S=(1+p)x_i$ being $x_i$ the observed inflection point, in agreement with Eq.\ \eqref{Sdef} and $p=0.5$. The results obtained by means of procedure (i) are shown in Table \ref{tab:Tabella3bis}: note that the estimated time instant has an absolute relative error (RAE) equal to $1.09\%$  for $\sigma=0.01$ and $1.00\%$ for $\sigma=0.02$. For the procedure (ii), as done in other similar works (cf.\ for instance, Di Crescenzo \textit{et al.}\ (2022) \cite{DiCrescenzoetal2022}), we use the \texttt{R} package \texttt{fptdApprox} (see Rom\'an-Rom\'an \textit{et al.}\ (2008) \cite{Romanetal2008} and \cite{fptdApprox})  to approximate the FPT density of the process through the boundary $S$.
					Figure \ref{fig:Figura9} shows the approximated FPT density and the FPT location function.
					\begin{figure}[h]
						\centering
						\subfigure[]{\includegraphics[scale=0.35]{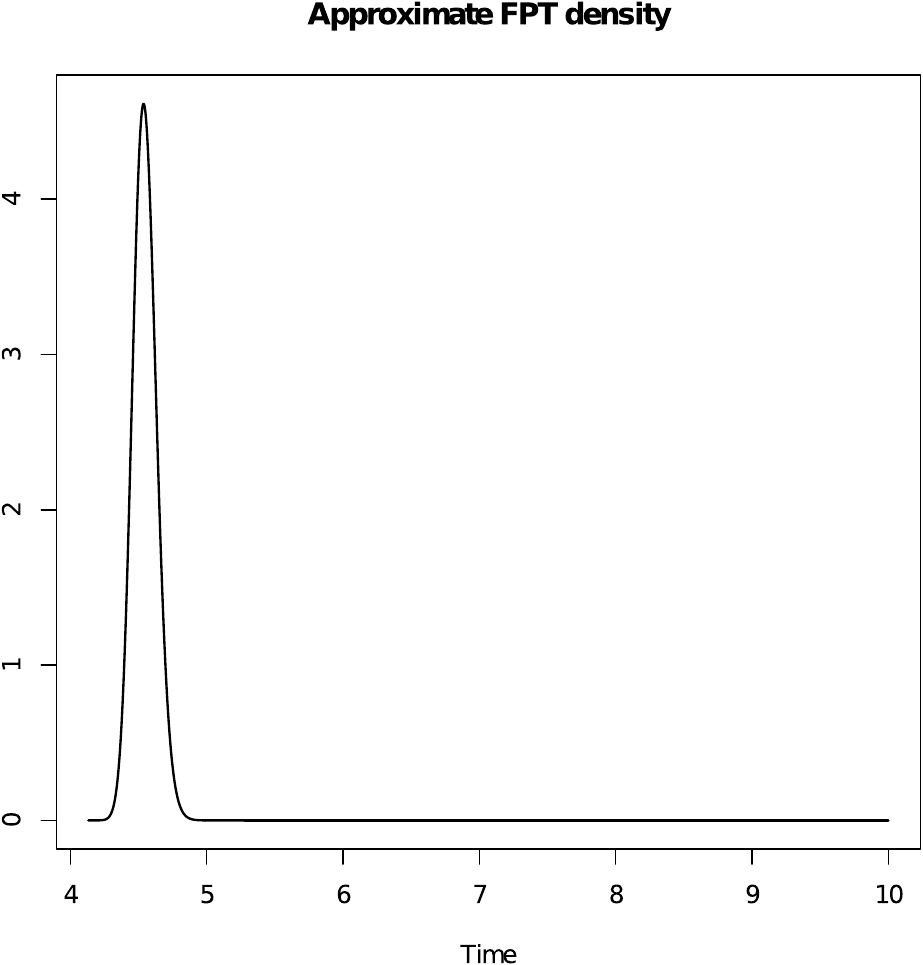}}\qquad
						\subfigure[]{\includegraphics[scale=0.35]{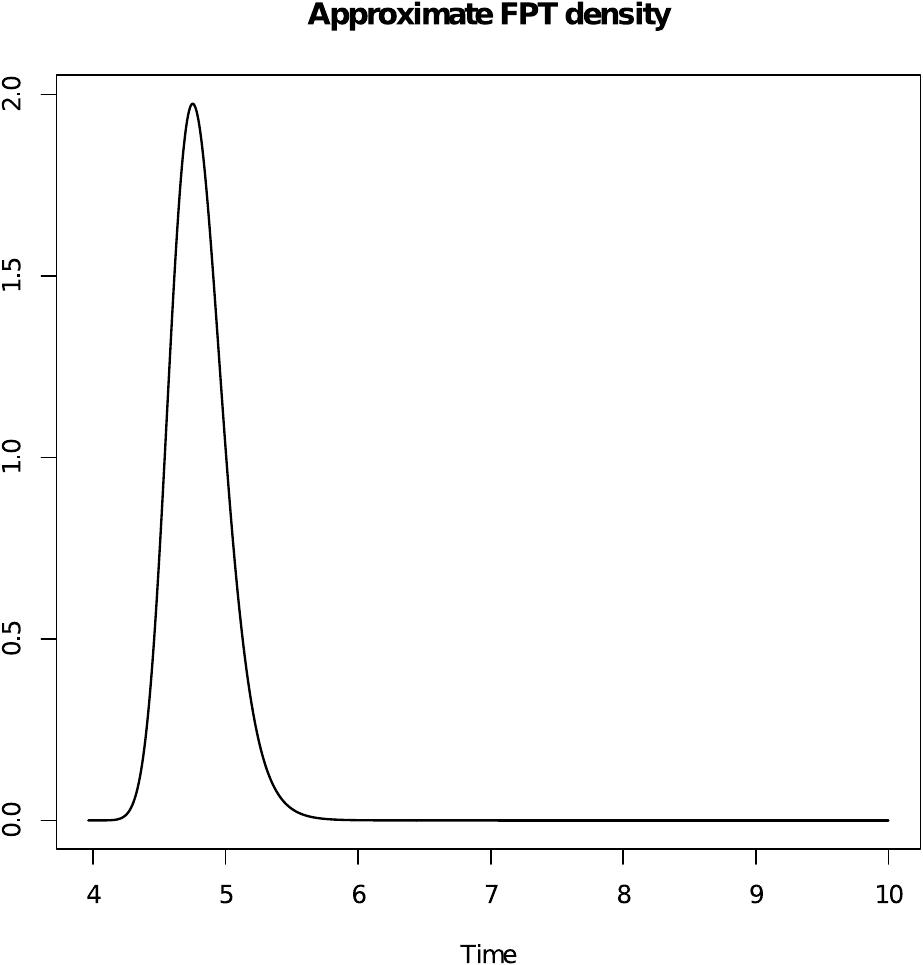}}\\
						\subfigure[]{\includegraphics[scale=0.35]{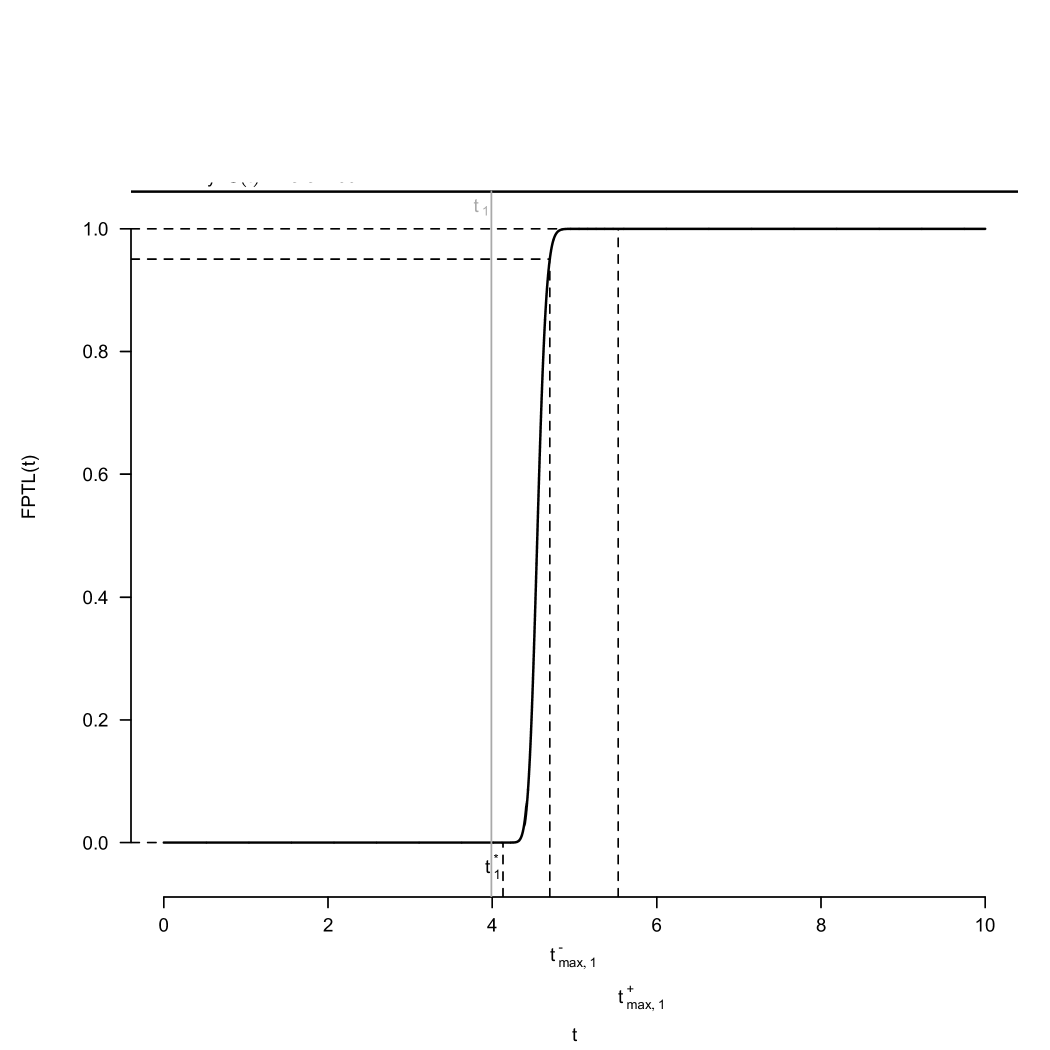}}\qquad
						\subfigure[]{\includegraphics[scale=0.35]{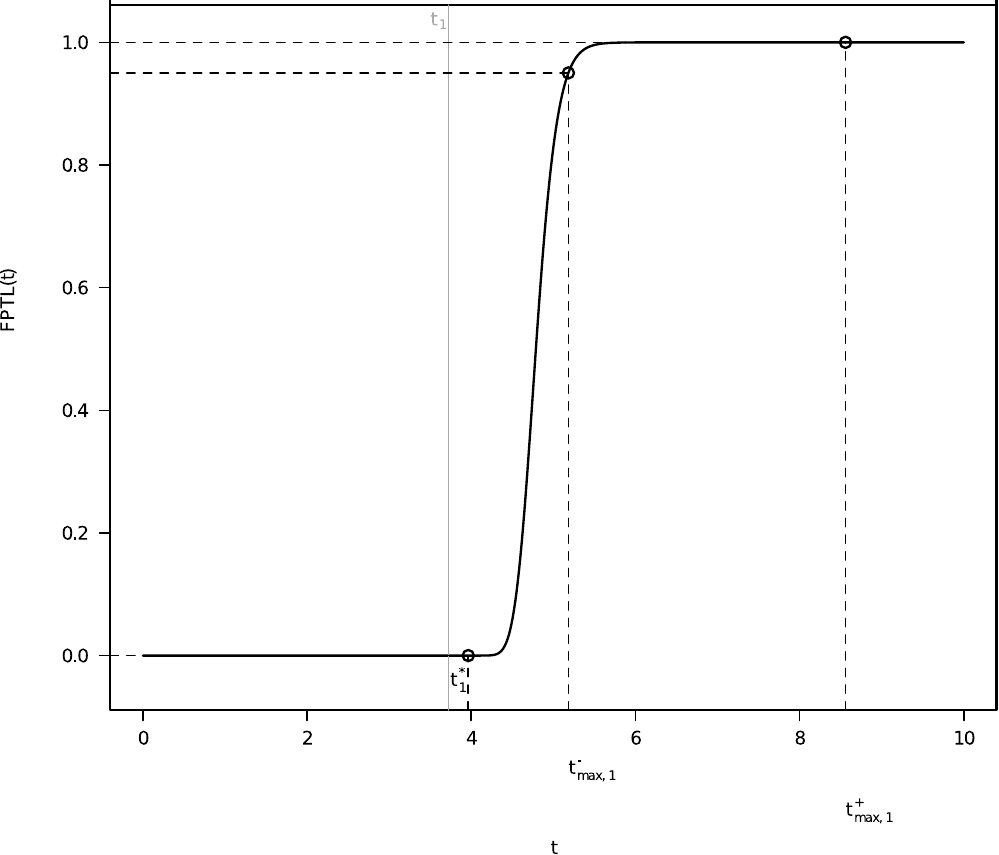}}
						\caption{The approximated FPT density function with the MLEs obtained by SA algorithm for (a) $\sigma=0.01$ and (c) $\sigma=0.02$. The FPT location function with the MLEs obtained by SA algorithm for (b) $\sigma=0.01$ and (d) $\sigma=0.02$.  }
						\label{fig:Figura9}
					\end{figure}
					In Table \ref{tab:Tabella3bis} we provide the mean, the standard deviation, the mode, the $1$st, the $5$th and the $9$th deciles of the FPT. In this case, the RAE between the estimated and the theoretical inflection time is $2.71\%$ for $\sigma=0.01$ and $8.45\%$ for $\sigma=0.02$.

					\begin{table}[]
						\caption{The theoretical value and the estimate obtained via Procedure (i) of the time instant $t^*$. The mean, the mode, the $1$st, the $5$th and the $9$th decile and the standard deviation (st.\ dev.) of the FPT of the approximated diffusion process through the boundary $S=(1+p)x_i$,  for $p=0.5$. The absolute relative error of the estimates is also provided. The quantities are obtained by using SA.
						}
						\label{tab:Tabella3bis}
						\tiny
						\centering
							\begin{tabular}{l|l|l|l|l|l|l|l|l|l}
								instant &
								th.\ val. &
								$\sigma$ &
								result of Proc.\ (i) &
								mean &
								mode &
								1st dec. &
								5th dec. &
								9th dec. &
								st.\ dev. \\ \hline
								\multirow{4}{*}{$t^*$} &
								\multirow{4}{*}{$4.42611$} &
								$0.01$ &
								$4.47456$ &
								$4.54597$ &
								$4.55862$ &
								$4.53713$ &
								$4.43811$ &
								$4.54388$ &
								$0.09363$ \\
								&  & (relative error:) & ($0.01095$) & ($0.02708$) & ($0.02994$) & ($0.02508$) & ($0.00271$) & ($0.02661$) & -         \\ \cline{3-10}
								&  & $0.05$            & $4.47060$   & $4.80032$   & $4.75344$   & $4.54871$   & $4.78756$   & $5.08003$   & $0.78458$ \\
								&  & (relative error:) & ($0.01005$) & ($0.08455$) & ($0.07395$) & ($0.02770$) & ($0.08166$) & ($0.14774$) & -
							\end{tabular}		
					\end{table}

					\subsection{Step 3: Estimate of $C(t)$}
					With the aim of estimating the function $C(t)$, we consider the mean of the stochastic process $X(t)$ (cf.\ Eq.\ \eqref{difproc1}) which corresponds to the classical deterministic curve $x_\theta(t)$ (cf.\ Eq.\ \eqref{BR}). The estimated function $\widehat C (t)$ has been obtained by means of Eq.\ \eqref{Cstim}, and it is plotted in Figure \ref{fig:Figure8}(a)-(b).  In particular, the plot shows a worsening of the estimation of $C(t)$ for large times, as expected.
					
					\begin{figure}[h]
						\centering
						\subfigure[]{\includegraphics[scale=0.35]{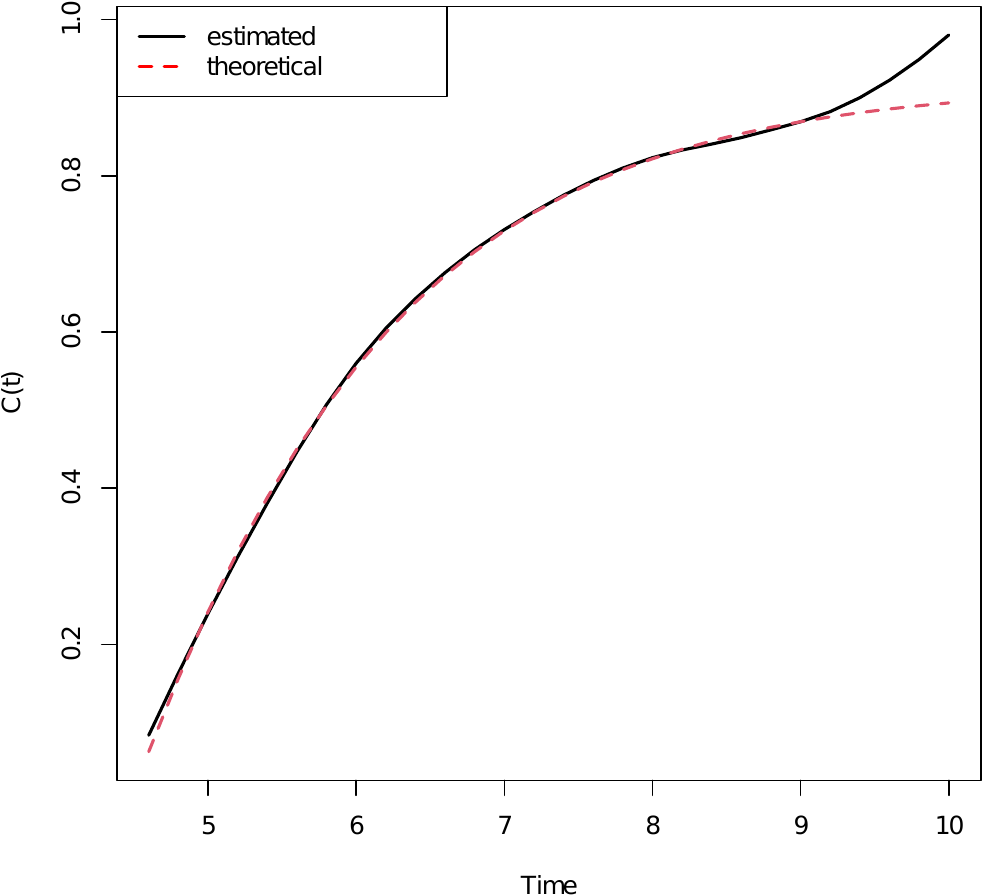}}\qquad
						\subfigure[]{\includegraphics[scale=0.35]{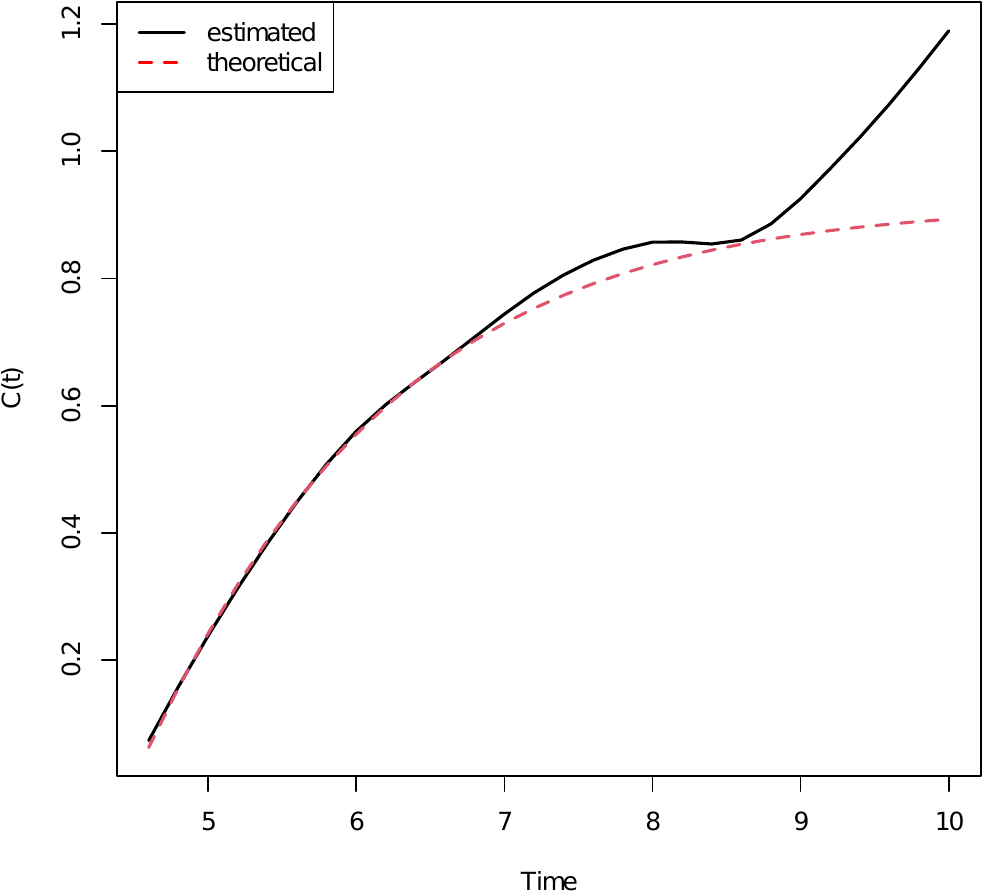}}\\
						\subfigure[]{\includegraphics[scale=0.35]{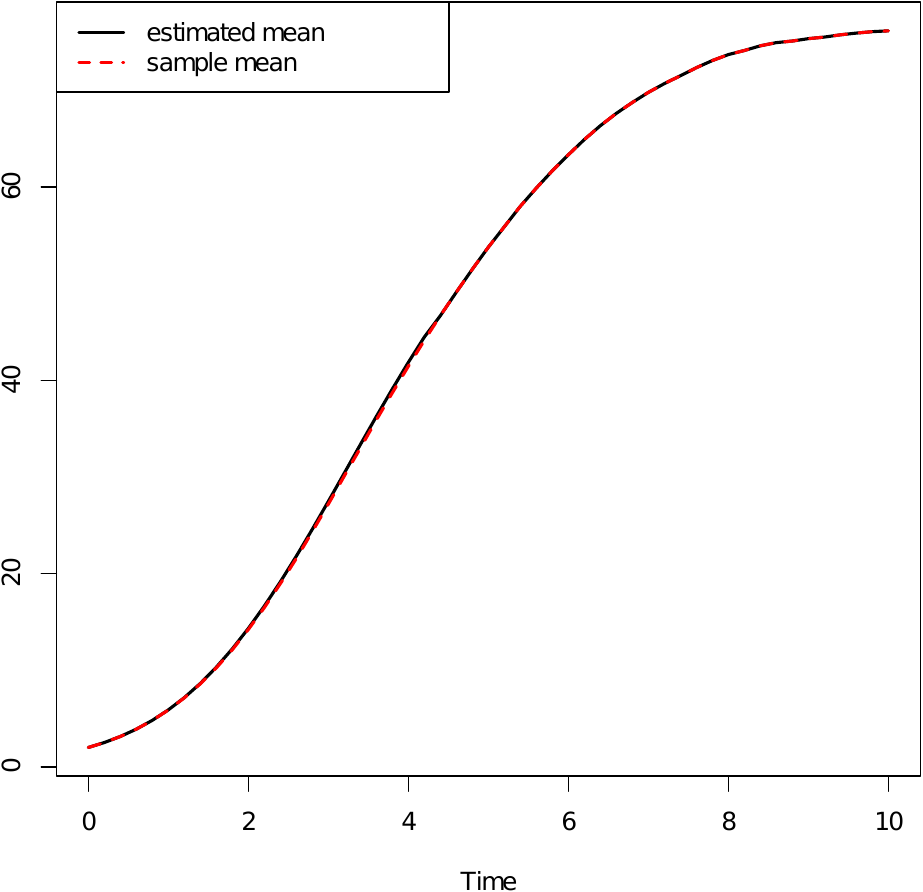}}\qquad
						\subfigure[]{\includegraphics[scale=0.35]{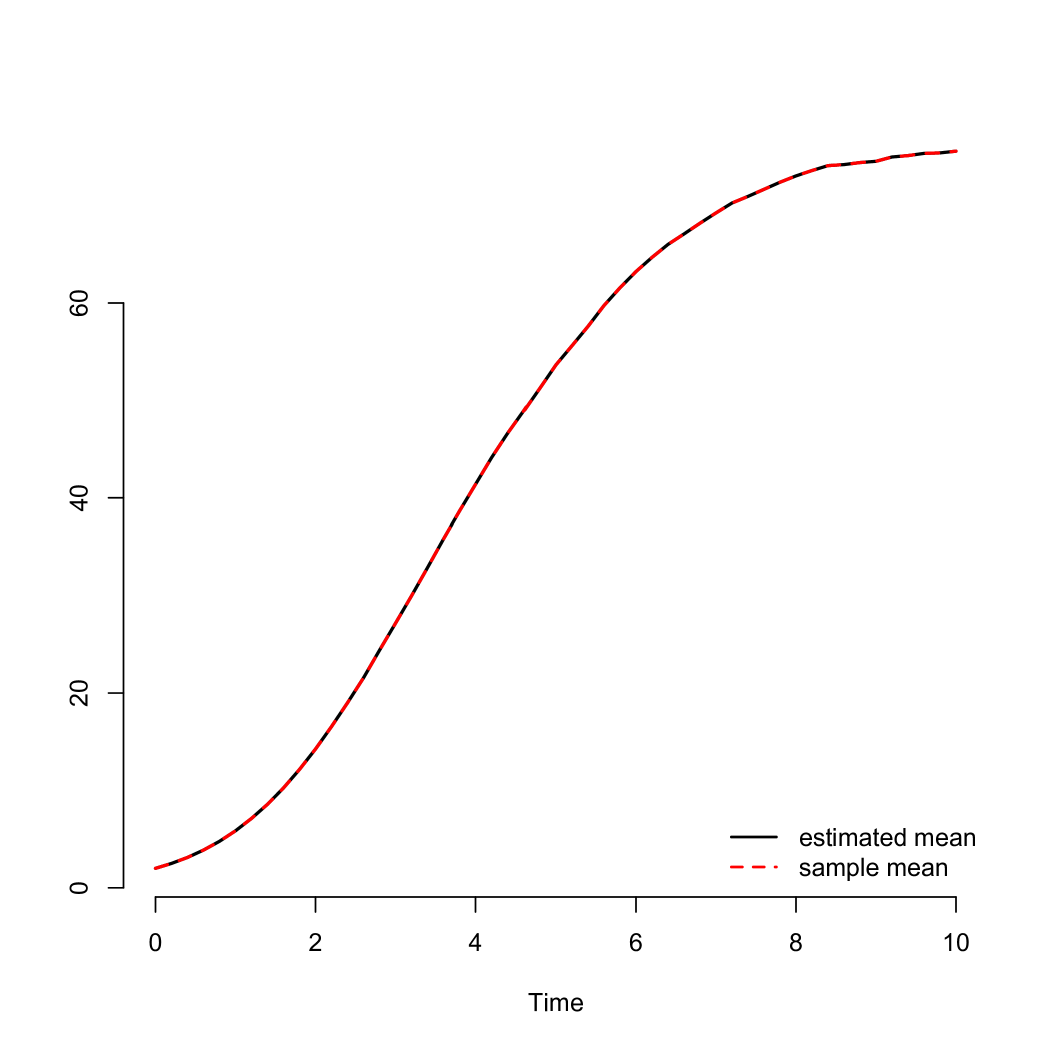}}\\
						\caption{The theoretical function $C(t)$ and the estimated function $\widehat C(t)$ for (a) $\sigma=0.01$ and (c) $\sigma=0.02$. The sample mean and the estimated mean for (b) $\sigma=0.01$ and (d) $\sigma=0.02$. }
						\label{fig:Figure8}
					\end{figure}
					In order to have a quantitative measure of the goodnees of its estimation, we have also computed the RAE between the theoretical function $C(t)$ and the estimated function $\widehat C(t)$, i.e.
					\begin{equation*}
						RAE=\frac{1}{N^*}\sum_{j=1}^{N^*}\frac{\left|C(t_j)-\widehat{C}(t_j)\right|}{C(t_j)},
					\end{equation*}
					where $N^*$ denotes the number of observation times between $\widehat t^*$ and $T$.
					When $\sigma=0.01$, we have $RAE\simeq 0.02101$, whereas if  $\sigma=0.02$, we find  $RAE\simeq0.05775$.
					
					Using the estimated function $\widehat C(t)$, we obtain the estimated mean  $\widehat{ \mathsf{E}} [\widetilde X (t)]$ of the modified process given in Eq.\ \eqref{difproc}. Both the estimated mean  and the sample mean are plotted in Figure  \ref{fig:Figure8}(c)-(d) and they almost coincide. Clearly, when $t\to T$ the difference between $C(t)$ and $\widehat C(t)$ is more relevant.
					\section{Application to real data of  oil production}\label{Sect9}
					The considered model (cf.\ Eq.\ \eqref{difproc}) is reasonable to describe phenomena in which the growth rate may be modified starting \textcolor{blue}{from} a specific instant by known external factors. A representative example will be considered in this section. It is concerning  oil production since external causes may be effective at a certain time in order to increase the amount of extracted oil.
					Hence, let us consider an application of the stochastic model introduced in Section \ref{DiffProc} to real data concerning oil production in France, taken from \cite{linkdata} and  reported in Figure \ref{fig:Figure16feb}(a).
					\begin{figure}[h]
						\centering
						\subfigure[]{\includegraphics[scale=0.4]{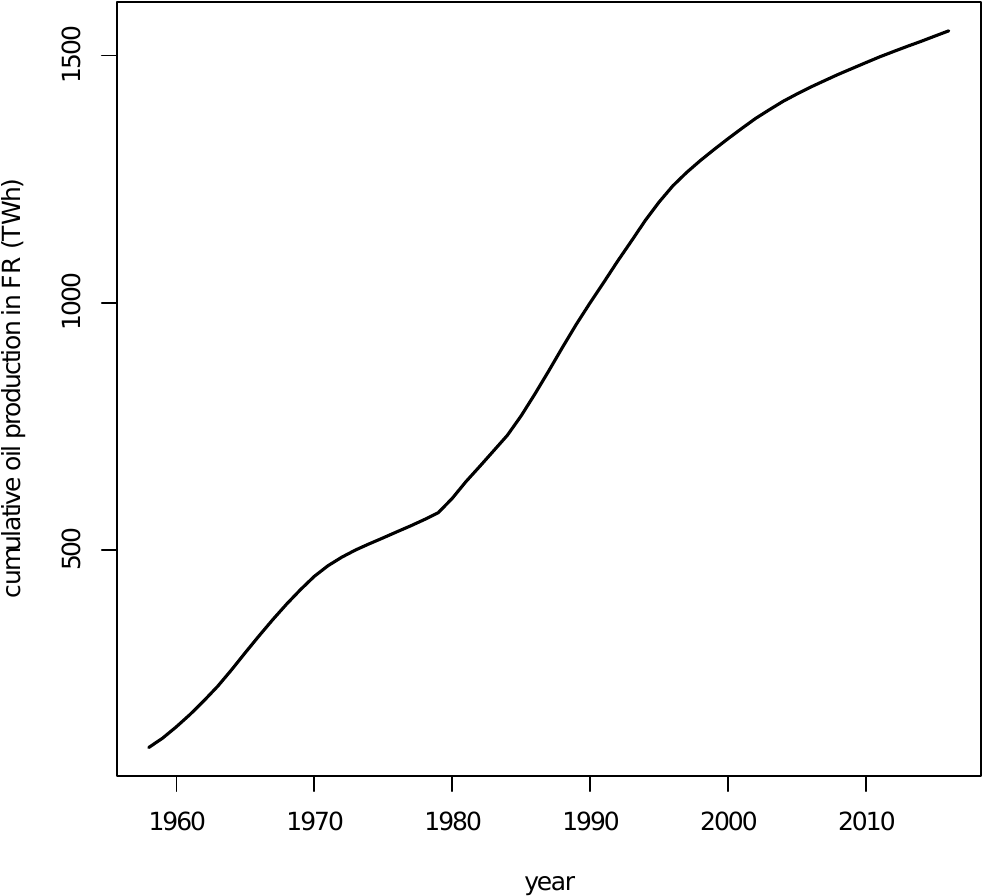}}\qquad
						\subfigure[]{\includegraphics[scale=0.4]{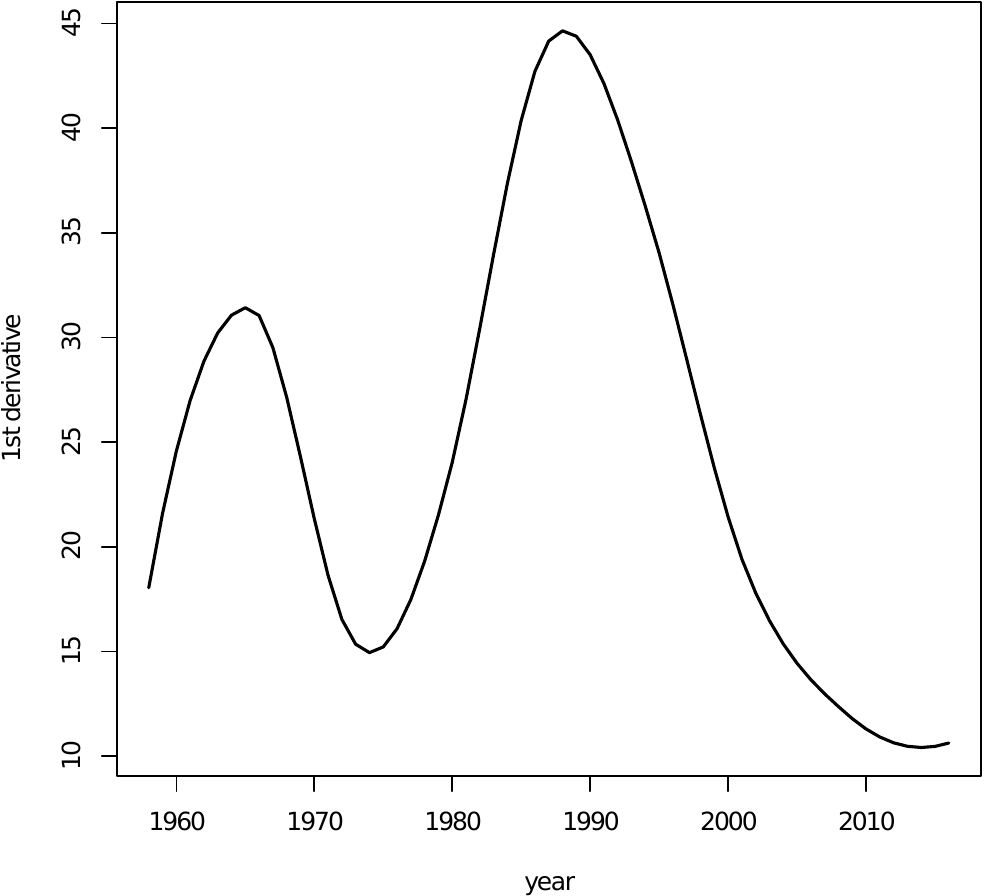}}\\
						\caption{(a) Cumulative oil production of France  and (b) 1st derivative of the cubic spline interpolating the data between $1958$ and $2016$. }
						\label{fig:Figure16feb}
					\end{figure}
					The amount of produced oil  is measured in TWh (Terawatt per hour). We consider the time interval from $1958$ (the first year in which the cumulative oil production exceeds $100$ TWh) to $2016$. In this case, the curve representing  the yearly oil production exhibits more than one peak, as in Laherrère (2000) \cite{Laherrere}. Our main goal is to establish that the considered model \eqref{difproc} is  in agreement with the observed real data by  following the steps of \textbf{Procedure 1}.\\
					
					\textbf{Step 1} As illustrated in Section \ref{Step1}, we approximate the data with a natural cubic spline  $S(t)$ and we determine an approximation of the inflection time instant $t_I$ by means of the derivatives of $S(t)$.  The derivative of $S(t)$  is \textcolor{blue}{plotted} in Figure \ref{fig:Figure16feb}(b).  The approximated inflection time instant is given by $t_I\simeq1964$. The intervals $I_\nu$ with $\nu\in\{q,k,\eta,\sigma\}$ and the corresponding MLEs are given in Table \ref{tab:Tabella16feb}.
					The MLEs have been determined by means of the data over the restricted interval $[t_0,\bar t]$, where $\bar t$ is the time instant corresponding to the boundary $S=(1+p)x_I^*$ being $x_I^*$ the observed inflection point and $p>0$. Unlike the examples examined in the simulation study (cf.\ Section \ref{Simulation}) where $p$ was assumed to be known, we now examine the case where $p$ is previously unknown. For estimate purposes, we consider a set of reasonable values for $p$. Among these, we select the value $\widehat p$ which finally minimizes the RAE between the estimated mean and the sample mean (see Step 3 below).\\
					\begin{table}[]
						\caption{The bounded intervals, their width and the MLEs (determined with SA) of the parameters considering different value of $\widehat p$.}
						\label{tab:Tabella16feb}
						\centering
						\begin{tabular}{l|l|l|l|l}
							$\widehat p$                    							& parameter                   	& interval                                                 						& width                     						&
							 MLE              \\ \hline
							\multirow{4}{*}{$0.3$}  		 	& {$q$}        						&{$I_q=[3.16349\cdot 10^{-2},2.28375]$}                  		&{$2.25211$} 			     			& $0.84526$         \\ \cline{2-5}
							&{$k$}        						&{$I_k=[4.71289\cdot10^{-8},0.94024]$}  		&{$0.94024$} 		    		 	& $0.78784$         \\ \cline{2-5}
							&{$\eta$}     					   &{$I_\eta=[0, 2.01897]$} 										&{$2.01897$} 		     			& $0.13325$         \\ \cline{2-5}
							&{$\sigma$} 					&{$I_{\sigma}=(0,0.1)$} 											&{$0.1$} 					 				& $0.00370$ \\ \hline
							\multirow{4}{*}{$0.5$}  		 	& {$q$}        						&{$I_q=[7.88791\cdot 10^{-2},6.04903]$}                  		&{$6.04114$} 			     			& $0.72390$         \\ \cline{2-5}
							&{$k$}        						&{$I_k=[1.11653\cdot10^{-3},0.95147]$}  		&{$0.95035$} 		    		 	& $0.76390$         \\ \cline{2-5}
							&{$\eta$}     					   &{$I_\eta=[0, 5.47622]$} 										&{$5.57622$} 		     			& $0.10121$         \\ \cline{2-5}
							&{$\sigma$} 					&{$I_{\sigma}=(0,0.1)$} 											&{$0.1$} 					 				& $0.00359$ \\ \hline
							\multirow{4}{*}{$0.7$}  		 &{$q$}        						  &{$I_q=[2.0322\cdot 10^{-1},5.43986]$}                 			&{$5.23664$} 				     			& $0.71189$         \\ \cline{2-5}
							&{$k$}       	 					&{$I_k=[1.63578\cdot 10^{-1},0.93236]$}                 					&{$0.76878$} 					    			& $0.76115$         \\ \cline{2-5}
							&{$\eta$}    	 					&{$I_\eta=[6.36869\cdot 10^{-3},4.40902]$}          &{$4.40265$} 					     			& $0.09807$         \\ \cline{2-5}
							&{$\sigma$} 					&{$I_{\sigma}=(0,0.1)$}  										& {$0.1$}   						     				& $0.00325$  \\
						\end{tabular}
					\end{table}

					\textbf{Step 2} We compute the estimated value of $t^*$ by means of the two available  procedures: (i) using the MLEs and compute $\hat t^*$ by means of the deterministic formula given in Eq.\ \eqref{tstar}, or (ii) computing  $\hat t^*$ as the mean of the FPT of the estimated process through $S=(1+\widehat p)x_I^*$ (see Section \ref{step2} for further details). The obtained results are given in Table \ref{tab:Tabella16feb2}. As can be noticed, for any value of $\widehat p$, the deterministic value of $t^*$ is different whereas the mean of the FPT is almost the same.
					\begin{table}[]
						\caption{The deterministic value of the time $t^*$, the mean, the mode, the $1$st, the $5$th and the $9$th decile and the standard deviation (st.\ dev.) of the FPT of the approximated diffusion process through the boundary $S$ for $t_0=1958$ and different values of $\widehat p$.}
						\label{tab:Tabella16feb2}
						\small
						\begin{tabular}{l|l|l|l|l|l|l|l|l}
							instant 								& $\widehat p$                       & det.\ val. 			& mean      			& mode      		& 1st dec. 			& 5th dec. 			& 9th dec. 			& st. dev. \\ \hline
							{$t^*-t_0$} 						&{$0.3$} 				& $10.59037$ 			& $8.24092$ 		& $8.23754$ 	& $8.11223$ 	& $8.24039$ 	& $8.37425$ 	& $0.10849$ \\  \cline{2-9}
							&{$0.5$} 				& $12.29745$ 			& $9.84548$ 		& $9.83844$ 	& $9.66999$ 	& $9.84610$ 	& $10.02986$ 	& $0.15087$ \\  \cline{2-9}
							&{$0.7$}    	     			& $15.68183$  		& $11.69182$ 		& $11.67821$ 	& $11.44741$  	& $11.68847$  	& $11.94491$  	& $0.20993$ \\
						\end{tabular}
					\end{table}
					We select the deterministic FPT to estimate $t^*$ for the results given in the simulation study  (cf.\ Section \ref{Simulation}). We point out that, for $\widehat p=0.5$, the specific value of the estimate $\hat t^*$ is $\hat t^*\simeq1970$ (by considering the deterministic estimate of $t^*$). This  estimate is in agreement with new oil explorations started in 70s as a consequence of the purpose of French government to invest in energy independence after the severe crisis emerged during those years (as reported in Lieber (1979) \cite{motivazioni}).\\
					
					\textbf{Step 3} Using the estimated function $\widehat C(t)$ obtained by means of Eq.\ \eqref{Cstim}, we compute the mean $\widehat {\textsf E} (\widetilde X(t))$ of the modified process.
					The estimate $\widehat p$ is selected by minimizing the RAE between the estimated mean and the sample mean. As can be deducted from Table \ref{tab:Tabella21feb}, in this case, we have $\widehat p=0.5$ which corresponds to the boundary $S=1.5 x_I^*$.
					\begin{table}[h]
						\centering
						\caption{The RAE between the estimated mean and the sample mean by considering different values of $\widehat p$.}
						\label{tab:Tabella21feb}
						\begin{tabular}{l|l|l|l}
							$\widehat p$         & $0.3$     & $0.5$    & $0.7$    	  \\ \hline
							$RAE$   & $0.00783\%$ & $0.00193\%$ & $0.04652\%$  \\
						\end{tabular}
					\end{table}
					In  Figure \ref{fig:FIgure16feba}, we show the estimated function
					\begin{equation}\label{integraledatireali}
					\int_{\widehat t^*}^{t} \widehat C(s) \frac{\widehat k^s |\log \widehat k|}{\widehat \eta+\widehat k^s}\mathrm {d} s, \qquad t\ge \widehat t^*
					\end{equation}
					the estimated mean and the sample mean of the data for $\widehat p=0.5$.
					\begin{figure}[h]
						\centering
						\subfigure[]{\includegraphics[scale=0.4]{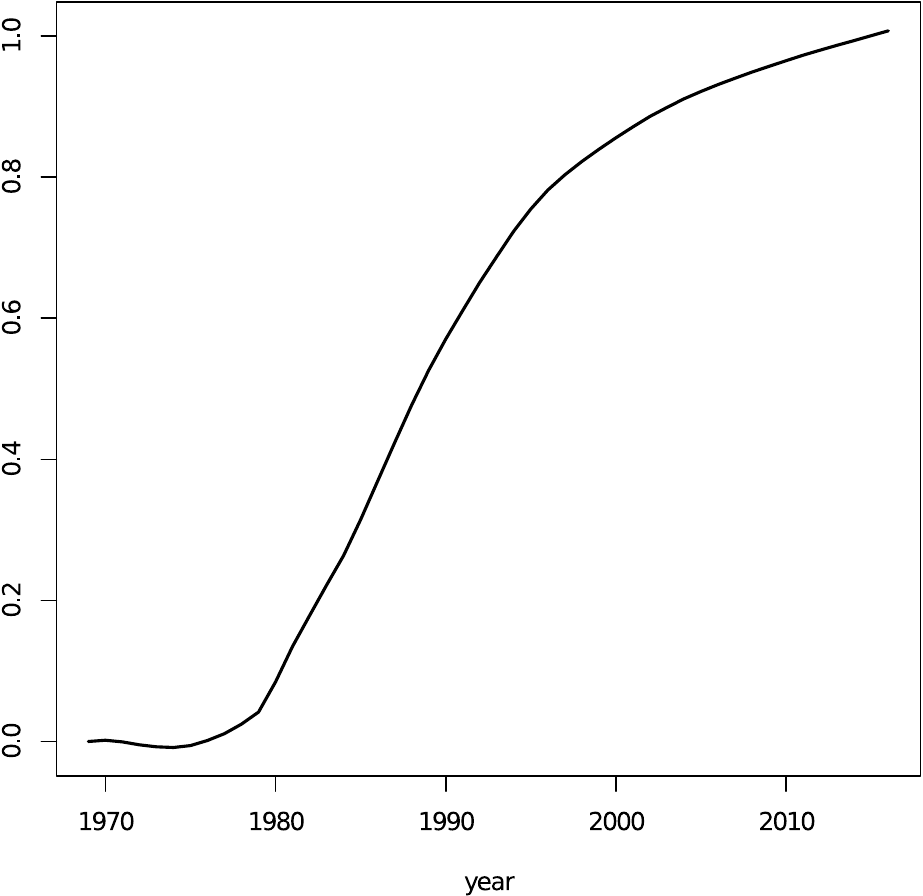}}\qquad
						\subfigure[]{\includegraphics[scale=0.4]{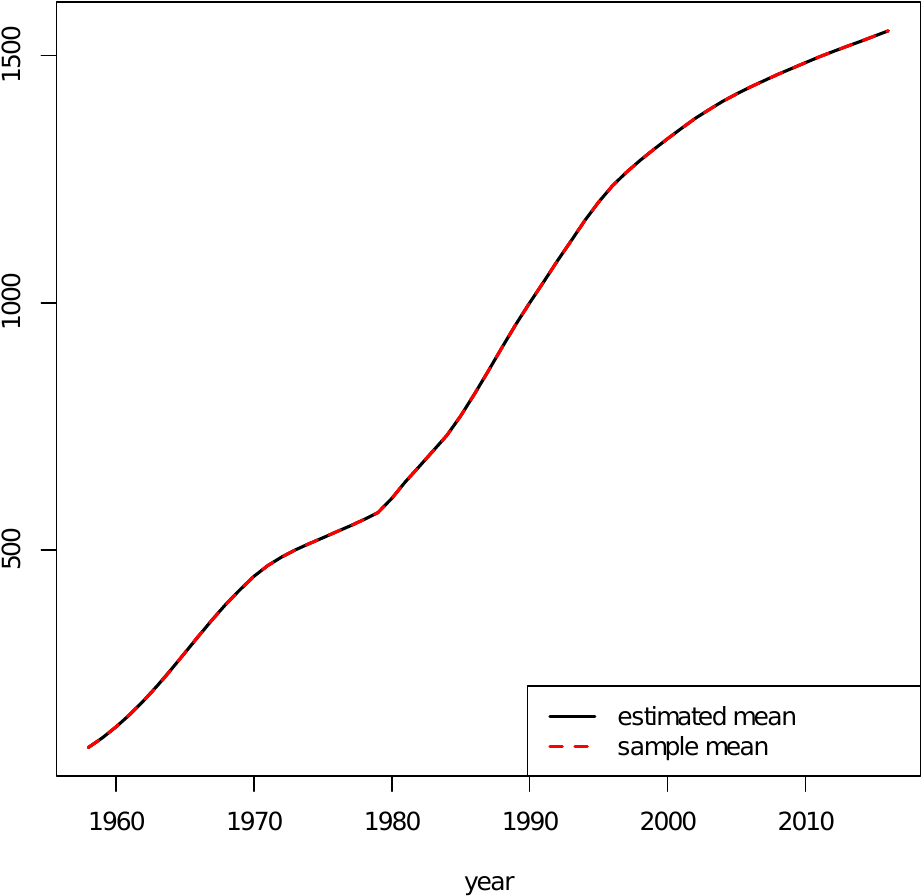}}\\
						\caption{(a) The estimated function given in Eq.\ \eqref{integraledatireali}, (b) the sample mean and the estimated mean for $\widehat p=0.5$ and parameters given in Table \ref{tab:Tabella16feb}.}
						\label{fig:FIgure16feba}
					\end{figure}
				The provided example of application confirms that the model \textcolor{blue}{introduced} in Section \ref{DiffProc} is appropriate for describing oil production
					  or other phenomena with a perturbed growth rate. Moreover, as seen in this section, for particular choices of the parameters the model may describe phenomena characterized by multiple inflection points.
					\section{Conclusions}
					During recent years many growth models have been introduced to describe real phenomena. One of them is the Richards curve, which is a generalization of the logistic function. In detail, the main difference between the Richards model and the logistic model is related to the ratio between the carrying capacity and the value of the curve at the inflection point. In the logistic case, this ratio is equal to $1/2$ whereas in the Richards case it is equal to $(1+1/q)^q$. Hence, the Richards curve seems to be more reasonable to describe phenomena in which the carrying capacity is not necessarily twice the value at the inflection point.
					In this work, we studied a special modification of the classical Richards model. Specifically, we substitute a parameter of the classical model \textcolor{blue}{with} a time-varying one aiming to describe situations in which the growth rate may be modified by external factors. The modification starts at a specific time instant which is identified as the first-crossing-time of the curve through a fixed boundary dependent on the value of the curve at the inflection point. The model has been described both from a deterministic and stochastic point of view. Two different kinds of stochastic processes have been introduced: birth-death processes and diffusion processes. The problems of parameters estimation and of the FPT have been also addressed. The determination of the MLEs has been conducted by means of suitable optimization methods. We considered gradient-free optimization algorithms since the expression of the derivative of the likelihood function is intricate, even if it is available in closed form. Furthermore, numerical methods have been used to determine approximations of the probability density function of the FPT of the modified process through constant boundaries.\\
					
					Regarding future developments, it may be interesting to study in more detail the function $C(t)$ and provide reasonable interpretation regarding its effects in real contexts. The model may be also enhanced to include the possibility that external factors cause a decrease of the growth rate.
					
					\section*{Acknowledgements}
					\textcolor{blue}{
					A.D.C.\ and P.P.\  are members of the research group GNCS of INdAM (Istituto Nazionale di Alta Matematica). This work was supported in part by the ``Ministerio de Ciencia e Innovaci\'on, Spain, under Grant PID2020-1187879GB-100, and ``Mar\'ia de Maeztu'' Excellence Unit IMAG, reference CEX2020-001105-M, funded by MCIN/AEI/ 10.13039/501100011033/, and by  the `European Union -- Next Generation EU' through MUR-PRIN 2022, project 2022XZSAFN ``Anomalous Phenomena on Regular and Irregular Domains: Approximating Complexity for the Applied Sciences'', and MUR-PRIN 2022 PNRR, project P2022XSF5H ``Stochastic Models in Biomathematics and Applications''.
				}

				\end{document}